\documentclass[10pt]{amsart}
\usepackage[english]{babel}
\usepackage{amsmath,amssymb,amsthm,amscd}

\usepackage{latexsym}
\usepackage{amstext}
\usepackage{amsgen,amsfonts,amsbsy}
\usepackage{url}
\usepackage{color,hyperref}
\usepackage{tikz-cd}

\theoremstyle{plain}
\newtheorem{theorem}{Theorem}[section]
\newtheorem{proposition}[theorem]{Proposition}
\newtheorem{lemma}[theorem]{Lemma}
\newtheorem{corollary}[theorem]{Corollary}

\theoremstyle{definition}

\theoremstyle{remark}%following is written with \rm

\newtheorem{remark}[theorem]{Remark}

\newcommand{\N}{\mathbb{N}}
\newcommand{\Z}{\mathbb{Z}}

\newcommand{\R}{\mathbb{R}}
\newcommand{\set}[2]{\{#1\,|\ #2\}}

\newcommand{\f}{\mathbf{v}}
\newcommand{\g}{\omega}

\newcommand{\C}{\mathbb{C}}

\newcommand{\GL}{\mathrm{GL}}
\newcommand{\Ad}{\mathrm{Ad}}

\newcommand{\SL}{\mathrm{SL}}
\newcommand{\Sp}{\mathrm{Sp}}

\newcommand{\U}{\mathrm{U}}

\newcommand{\ket}[1]{\vert\mathit{#1}\rangle}

\newcommand{\mat}[4]{\left(\begin{array}{cc}
 #1 & #2\\
 #3 & #4\\
\end{array}\right)}

\newcommand{\matN}[4]{\left[\begin{array}{cc}
 #1 & #2\\
 #3 & #4\\
\end{array}\right]_N}

\newcommand{\vc}[2]{\left(\begin{array}{c}
 #1 \\
 #2\\
\end{array}\right)}

\begin{document}
\title[Clifford group is not a semidir. product if  $N$ is divisible by four]{Clifford group is not a semidirect product in dimensions $N$
divisible by four}

\author[M.~Korbel\'a\v{r}]{Miroslav~Korbel\'a\v{r}}
\address{Department of Mathematics, Faculty of Electrical Engineering, 
Czech Technical University in Prague, Jugosl\'avsk\'ych partiz\'an\r{u} 1, 
166 27 Prague 6, Czech Republic}
\email{korbemir@fel.cvut.cz}

\author[J.~Tolar]{Ji\v{r}\'i Tolar}
\address{Department of Physics,
 Faculty of Nuclear Sciences and  Physical  Engineering,
 Czech Technical University in Prague, B\v{r}ehov\'{a} 7,
 115 19 Prague 1, Czech Republic}
\email{jiri.tolar@fjfi.cvut.cz}

\thanks{}

\subjclass[2010]{}
\date{\today}

\begin{abstract}
The paper is devoted to projective Clifford groups of quantum $N$-dimensional systems. 
Clearly, Clifford gates allow only the simplest quantum computations
which can be simulated on a classical computer (Gottesmann-Knill theorem).
However, it may serve as a cornerstone of full quantum computation.
As to its group structure it is well-known that -- in $N$-dimensional 
quantum mechanics -- the Clifford group is a natural semidirect product
provided the dimension $N$ is an odd number.
For even $N$  special results on the Clifford groups are scattered 
in the mathematical literature, but they don't concern the semidirect structure.
Using appropriate group presentation of $\SL(2,\Z_N)$ 
it is proved that for even $N$ projective Clifford groups are not 
natural semidirect products if and only if $N$ is divisible by four.
\end{abstract}
\maketitle
%\vspace{4ex}

\section{Introduction}

Finite-dimensional quantum mechanics has naturally become 
the mathematical as well as physical basis of quantum information. 
Among its main concepts, Clifford groups belong to the basic 
symmetries of finite-dimensional quantum kinematics and 
have over time emerged as an important tool in
quantum information with a wide variety of applications.

The term Clifford group was coined in \cite{Gottesman} in the study 
of fault-tolerant quantum computation and the construction 
of quantum error-correcting codes. 
The simplest Clifford group in multiqubit quantum computation is generated 
by a restricted set of unitary Clifford gates -- 
the Hadamard, $\pi/4$-phase and controlled-X gates. 
Because of this the Clifford model of quantum computation 
can be efficiently simulated on a classical computer 
(the Gottesmann-Knill theorem) \cite{NielsenChuang}.
However, this fact does not diminish the importance of the Clifford
model, since it may serve as a suitable starting point for 
full-fledged quantum computation.
\footnote{We would like to point out that the claim in 
	\cite[note 43]{Gross} that the Clifford group 
which appears in the context of quantum information has no connection to 
Clifford groups used in the representation theory of orthogonal Lie groups, 
deserves a correction. Namely, the original Clifford
groups are certain automorphisms of the Clifford algebras
which appeared in a `square-root' problem of matrix linearization 
of a quadratic form  (like Dirac's equation). 
Looking back, the paper by Morinaga and Nono published in 1952
\cite{MorinagaNono} should be quoted.
There a more general linearization of a form of $n$-th degree was 
studied and shown to lead to what they called generalized Clifford algebras.
Then the terms Clifford transform group and Clifford collineation group 
were used in \cite{bolt1,bolt2} for what we now call
the Clifford group and the Heisenberg group, respectively.
}

In the general case of a quantum system with
$N$-dimensional Hilbert space the Heisenberg group of
unitary operators defines the quantum kinematics and the states of
the quantum register in quantum computing. 
Then the corresponding Clifford group $C(N)$ is defined as the group 
of unitary operators that map  the Heisenberg group, generated by scalar multiples of 
generalized Pauli matrices, to itself under conjugation.
 In other words, $C(N)$ is the normalizer of the Heisenberg group within the unitary group $\U(N)$.
It is clear that $C(N)$  necessarily contains
the continuous group $\U(1)$ of phase factors.
In quantum computing the so-called Clifford operations, i.e.
the corresponding operator transformations, are not affected
by phase factors. Therefore it is natural to study the 
\textit{projective Clifford group} -- the quotient group 
$\overline{C}(N)$ which counts classes of unitaries
in $C(N)$ that differ by an overall phase factor.
% This projective Clifford group 
%was called the Clifford quotient group in \cite{TolarJPCS18}.

We would also like to draw attention to the fact that we published 
comprehensive studies of the projective Clifford groups
 $\overline{C}(N)$ under the name \textit{symmetries of finite Heisenberg groups}
\cite{HPPT02, KorbTolar10, KorbTolar12}.
This terminology stems from our research in the setting of a
programme envisaged by J. Patera and H. Zassenhaus -- the study of symmetries 
in the context of Cartan's theory of semi-simple Lie algebras: 
especially the symmetries of the Pauli gradings of 
the Lie algebras $\mathrm{sl}(N,\C)$ \cite{HPP98, HPPT02, PZ88, PZ89}
for arbitrary dimensions $N$.

It is therefore evident that we were motivated to study symmetries 
of finite Heisenberg groups not in prime or prime power dimensions, 
but for arbitrary dimensions. Note that our results 
\cite{HPPT02, KorbTolar10, KorbTolar12} cover a broad collection 
of projective Clifford groups corresponding to arbitrary single
or composite quantum systems. For odd dimensions see also \cite{TolarJPCS18,TolarChadz}. 
The papers \cite{Digernes, DigernesVV, Hostens} deserve to be mentioned, too.

In analogy with the continuous case -- remember that in quantum optics
squeezed states can be viewed as coherent states associated with
unitary irreducible representations of the inhomogeneous symplectic group --
it is natural to ask whether the Clifford group in finite dimensions should be
isomorphic to a semi-direct product of the finite symplectic group
with the Heisenberg group. Indeed, in odd dimensions it was really proved 
that both the projective and non-projective Clifford groups
have the structure of a \textit{natural semidirect product} 
 \cite[Lemma 4]{Gross}, \cite{appleby,DuttaPrasad}.

But if the dimension is an even number, the situation is more tricky.
To our best knowledge, special results are scattered in the literature
under many different guises, but the problem of semidirect product
structure seems to be unsettled yet. 
One partial result is due to D. M. Appleby \cite[Theorem 1]{appleby} 
disclosing that in even dimensions  the projective Clifford group 
may not be a natural semidirect product.
The part of Appleby's theorem \cite[Theorem 1]{appleby}  concerning even dimensions is taken
as our starting point in Section \ref{section_6}  where an appropriate 
presentation of $\SL(2,\Z_N)$ is chosen \cite{hopf}.
Using it for even $N$, in Sections \ref{section_6} -- \ref{section_12} new results 
for projective Clifford groups are derived: in Section \ref{section_12} we show
that for dimensions divisible by 4 they are not semidirect products,
but for dimensions $N$ with $N/2$ odd they are. Existence of the semidirect product in even dimensions seems to be quite surprising and we also propose a conjecture relating to this result in Conclusions (Section \ref{conclusion}). 
%We mainly refer to Appleby's papers 
%\cite{appleby,Appleby11} on which our argumentation is built.
\footnote{Note the recent paper \cite{Raussendorf} concerned with
quantum dynamics in phase space in both odd and even dimensions.}

Our decision to study rather the projective Clifford groups than the non-project\-ive ones was motivated by a much simpler description of the projective case via Appleby's result and also by the fact that the non-existence of the semidirect product in the projective case immediately implies the same conclusion for the non-projective case (see Corollary \ref{6.2}). Nevertheless, let us point out that the approach using group presentations can be further applied to decide whether also the non-projective case (for $N$ even with $N/2$ odd) admits the structure of the semidirect product.

% Unfortunately, our approach is not suited to extend our results to 
% non-projective Clifford groups in full generality.
% ons.}

The paper is organized as follows:
in Sections \ref{section_2} -- \ref{semidirect} the main notions are introduced 
on which our chain of reasoning in Sections \ref{section_6} -- \ref{section_12} is based.

\section{Basics of quantum mechanics}\label{section_2}

In  classical mechanics, the symmetries of Hamilton equations 
are constituted by canonical transformations. 
But L. van Hove showed in 1951 \cite{Tilgner,VanHove} that
the group of canonical transformations of phase space $\mathbb{R}^{2n}$
cannot be represented by quantal unitary transformations. 
Only special subgroups of canonical transformations admit unitary representations. 
One such important case are the linear canonical transformations
of the phase space $\mathbb{R}^{2n}$ preserving the Poisson
brackets. They form the \textit{symplectic group} $\Sp(2n,\mathbb{R})$
which for one degree of freedom reduces to
$\Sp(2,\mathbb{R})\cong \SL(2,\mathbb{R})$ with the action on the
Cartesian coordinates $q$ and momenta $p$
\begin{equation}
(q',p') = (q,p)\mathbb{A}=(q,p) \left( \begin{array}{cc}
a & b \\
c & d
\end{array}  \right), \quad ad-bc=1,
\end{equation}
i.e. $\mathbb{A} \in \SL(2,\mathbb{R})$.
These symplectic transformations together with translations of the phase space 
form a semidirect product group of inhomogeneous symplectic transformations,
also called symplectic affine transformations \cite{bolt1, Gross}.

In \textit{standard} quantum mechanics the Cartesian coordinates and momenta have
their quantum counterparts, the canonical observables $\hat{Q}_i$ and
$\hat{P}_j$. According to the \textit{Stone-von Neumann uniqueness theorem}, 
these observables -- satisfying the canonical commutation relations
$[\hat{Q}_k,\hat{P}_\ell]=i\hbar I \delta_{k\ell}$ --
are represented by essentially self-adjoint operators unitarily
equivalent to the operators of the Schr\"odinger representation in
the Hilbert space $L^2(\mathbb{R}^n,d^{n}q)$.
In the simplest case of the Hilbert space   $L^2(\mathbb{R}, dq)$ 
($n=1$) the operators are as follows
$$ \hat{Q}\psi(q)=q\psi(q), \qquad
\hat{P}\psi(q)=-i\hbar \frac{\partial\psi(q)}{\partial q}.$$
The operators $\hat{Q}$ and $\hat{P}$ generate one-parameter (abelian) groups
of unitary operators
$$ U(t)= e^{\frac{i}{\hbar} \hat{P} t }, \qquad
V(s) = e^{\frac{i}{\hbar} \hat{Q} s }
$$
satisfying the relation
$$ U(t)V(s)= e^{\frac{i}{\hbar}st} V(s)U(t). $$
In the Schr\"odinger representation
they act on $L^2(\mathbb{R}, dq)$ as shift and phase operators 
$$
[U(t)\psi](q) = \psi (q+t), \qquad
[V(s)\psi](q) =e^{\frac{i}{\hbar}qs} \psi(q)
$$
They are jointly expressed as unitary \textit{Weyl displacement operators} 
\begin{equation}  \label{W}
W(s,t) = e^{\frac{i}{\hbar}(s\hat{Q} + t\hat{P})} =
e^{\frac{i}{2\hbar}st} V(s)U(t).
\end{equation}
From the composition rule
\begin{equation}\label{composition}
W(s_1,t_1)W(s_2,t_2) = e^{-\frac{i}{2\hbar}(s_1 t_2 - s_2 t_1)}
W(s_1 + s_2, t_1 + t_2)
\end{equation}
one sees that they provide a \textit{projective unitary representation}
of the abelian group of translations of the phase space $\mathbb{R}^{2}$.

The Segal-Shale-Weil representation, shortly the \textit{Weil
representation}, is the projective unitary representation of the group
of linear canonical transformations in the same Hilbert space
$L^2(\mathbb{R}^n,d^{n}q)$. It can be lifted to a (non-projective) unitary irreducible
representation $X(\mathbb{A})$ of the double covering of
$\Sp(2n,\mathbb{R})$ (called the \textit{metaplectic group}) which
for $n=1$ reduces to the double covering of $\SL(2,\mathbb{R})$,
\begin{equation}
(\hat{Q}',\hat{P}') =(\hat{Q},\hat{P})\mathbb{A}, \quad
\det{\mathbb{A}}=1.
\end{equation}
The transformed commutators have the form of canonical commutation
relations $[\hat{Q}',\hat{P}']=i\hbar I$. Hence 
the Stone-von Neumann theorem implies the unitary equivalence
$$ \hat{Q}' = X(\mathbb{A})\hat{Q}X(\mathbb{A})^{-1}, \quad
\hat{P}' = X(\mathbb{A})\hat{P}X(\mathbb{A})^{-1}.
$$
Then the Weyl displacement operators are transformed as
$$ X(\mathbb{A})W(s,t)X(\mathbb{A})^{-1}=W((s,t)\mathbb{A}^T).
$$

\section{Heisenberg groups of $N$-dimensional quantum systems}\label{section_3}

The apparatus of finite-dimensional quantum mechanics can be obtained 
by the discretization of the standard quantum mechanics in Cartesian
coordinates. The introduction of its basic objects is guided by the
simplest quantum kinematics on the real line. 
The discrete objects for a simple $N$-level quantum system (a qudit)
are obtained by replacing the configuration space $\R$ by $\Z_N$  
\cite{Schwinger,StovTolar84,Vourdas,Weyl}.

In the following, the Hilbert space will be finite-dimensional, 
$\C^N$ with the standard inner product. 
%It is the state space of a quantum $N$-level system. 
In the usual notation of quantum computing \cite{appleby, Appleby11,NielsenChuang}, 
the basic unitary operators generalizing Pauli matrices 
$\sigma_x$ and $\sigma_z$ to dimension $N$ 
are denoted $X_N$, $Z_N$ (in our papers \cite{HPPT02,KorbTolar10,KorbTolar12} we denoted them  $P_N$, $Q_N$). 
%They generate the \textit{Pauli group} as a subgroup of the unitary group $\U(N)$ . 

Now in the Hilbert space $\mathcal{H}_{N}=\mathbb{C}^N$ take
an orthonormal basis $$\mathcal{B} = \left\lbrace\ket{0}, \ket{1},
\ldots \ket{N-1}\right\rbrace.$$ The basic unitary operators $Z_N$,
$X_N$ are defined by their action on the basis \cite{Schwinger,Weyl}
\begin{eqnarray}
Z_N \ket{j}&=& \omega_N^j \ket{j}, \\
X_N \ket{j} &=& \ket{j-1 \pmod{N}},
\end{eqnarray}
where $j=0,1,\ldots,N-1$, $\omega_N = e^{\frac{2\pi i}{N}}$. 
This is the well-known \textit{clock-and-shift representation} of the basic
operators $Z_N$, $X_N$. In the canonical or computational basis
$\mathcal{B}$ these operators are represented by the
\textit{generalized Pauli matrices}
\begin{equation}
Z_N = \mbox{diag}\left(1,\omega_N,\omega_N^2,\cdots,\omega_N^{N-1}\right)
\end{equation}
and
\begin{equation}
X_N = \left(
\begin{array}{cccccc}
0 & 1 & 0& \cdots & 0 & 0 \\
0 & 0& 1&  \cdots & 0 & 0 \\
0 & 0 & 0&\cdots & 0 & 0 \\
\vdots &   & & \ddots &   & \\
0 & 0 &0 & \cdots & 0 & 1 \\
1 & 0 &0 &\cdots & 0 & 0
\end{array} \right) .
\end{equation}
Their commutation relation
\begin{equation}
X_N Z_N = \omega_N Z_N X_N 
\end{equation}
expresses their minimal non-commutativity. Further, they are of the order $N$,
$$
 X_N^N = Z_N^N = I, \quad \omega_N^N = 1.
$$
Elements of $\mathbb{Z}_N = \left\lbrace 0,1,\ldots N-1
\right\rbrace $ label the vectors of the canonical basis
$\mathcal{B}$ with the physical interpretation that $\ket{j}$ is the
(normalized) eigenvector of position at $j \in \mathbb{Z}_N$. In
this sense the cyclic group $\mathbb{Z}_N$ plays the role of the
\textit{configuration space} for an $N$-dimensional quantum system. 

The \emph{Weyl  displacement operators} for the discrete case have an analogous form as in the continuous one (\ref{W}), in particular they are defined according to \cite{appleby} as
$$W(k,\ell)=\tau^{k\ell}X_N^k Z_N^\ell,
$$ 
where $\tau=-e^{\frac{i\pi}{N}}$  and  $k,\ell\in \Z$. \footnote{For $N$ odd we have $(X_N^iZ_N^j)^N=I$ for all $i,j\in\Z$ while for $N$ even, for instance, the operator $X_N Z_N$ is of order $2N$.
Such peculiarities of even dimensions suggest the choice of Weyl operators involving $\tau$ as a square root of $\omega_N$. Then $W(i,j)^N=I$ for all $i,j\in\Z$.

Notice further that $W(k+N,\ell)=W(k,\ell)=W(k,\ell+N)$ if $N$ is odd, while $W(k+N,\ell)=(-1)^{\ell}W(k,\ell)$ and $W(k,\ell+N)=(-1)^kW(k,\ell)$ if $N$ is even.} 

% {\color{red}
% By using the concept of Weyl displacement operators we would like to have for the discrete case an analogous form as in the countinous one (\ref{W}), in particular we require the form $$W(k,\ell)=\lambda^{k\ell}Z_N^k X_N^\ell$$ where $0\neq\lambda\in\C$ is a fixed number and  $k,\ell=0,1,\dots,N-1$. If we are further asking for an formula analogous to (\ref{composition}) 
% $$W(k_1,\ell_1)W(k_2,\ell_2)=\mu^{k_1\ell_2-k_2\ell_1} W(k_1+k_2,\ell_1+\ell_2)$$
% holding for a fixed number $0\neq\mu\in\C$ and all $k_1,k_2,\ell_1,\ell_2=0,1,\dots,N-1$ it turns out that the only possible choice is $\mu^{-1}=\lambda=\pm\sqrt{\omega_N}$ (see the Attachment).
% 
% Now the \emph{Heisenberg group} is defined as the smallest group containing $W(k,\ell)$ for all $k,\ell=0,1,\dots,N-1$. Denote now $\tau_N=\sqrt{}$
% }

As the quantum state is determined up a scalar multiple (the phase) it is natural to consider unitary operators (that perform transformation of states) also with scalar multiples. We define therefore an  equivalence relation $\sim$ on $\U(N)$ such that for $U,V\in \U(N)$ is $U\sim V$ if and only if $U=\lambda V$ for some $\lambda\in \U(1)$. The cosets of $\sim$ will be denoted as $U_{/\sim}$.

Motivated by this view the \textit{Heisenberg group} $H(N)$ is defined as the subgroup in $\U(N)$ generated by the operators 
$X_N$, $Z_N$ and all possible scalar multiples:
\begin{equation}
H(N)= \left\lbrace \lambda W(k,\ell) |
\lambda\in \U(1), k,\ell = 0,1,2,\ldots,N-1\right\rbrace.
\end{equation}
% The center $Z(H(N))$ of the Heisenberg group is the
% set of all those elements of $H(N)$ which commute with all elements in $H(N)$, hence $Z(H(N))=U(1)$.
% Since the center is a normal subgroup, one can go over to the
% quotient group
This group is also sometimes called \emph{generalized Pauli group} \cite{appleby}.

Since the group $I(N)=\set{\lambda I}{\lambda\in \U(1)}$ is a normal subgroup of $H(N)$, one can go over to the
 quotient group
$\overline{H}(N)=H(N)/I(N)$ that is the \emph{projective Heisenberg group}.
This quotient group is usually identified with the 
\textit{finite phase space} $\mathbb{Z}_N \times \mathbb{Z}_N$
by the isomorphism of abelian groups $$\nu:\Z_N \times \Z_N \rightarrow \overline{H}(N),\ \text{where}\ (k,\ell)\mapsto W(k,\ell)_{/\sim}\ .$$

 %Its elements are the cosets
%labeled by pairs of exponents $(i,j)$, $i,j = 0,1,\ldots,N-1$. 

% The \textit{Pauli group} $\Pi_N$ of order $N^3$ is generated by
% $Z_N$ and $X_N$ and consists of the unitary operators \cite{HPPT02,PZ88}:
% \begin{equation}
% \Pi_N = \left\lbrace \omega_N^l Z^i_N X_N^j |
% l,i,j = 0,1,2,\ldots,N-1\right\rbrace.
% \end{equation}
% Then the Clifford group $C(N)$ to be constructed should contain the
% Pauli group as its normal subgroup. 
% In the discretization of basic quantum operators 
% special role is played by the Weyl displacement operators (\ref{W}) 
% which involve the factor $1/2$ in the exponent. The continuous exponentials are
% replaced by  $\Z_N$-powers of $\omega_N$ and one observes
% that the division by $2$ is possible in $\Z_N$ for odd $N$ only. 
% In this way the discrete Weyl displacement operators
% may be guessed for odd $N$ from their form in the phase space $\R^2$.
% The continuous and the discrete odd case behave very similarly
% inspite of different mathematics \cite{Gross}. 

\section{Clifford groups and projective Clifford groups}\label{section_4}

The \emph{Clifford group} $C(N)$ comprises all unitary operators $U$ such that their action of inner automorphisms leaves the group $H(N)$ invariant:
$$ C(N)=\set{U\in \U(N)}{UH(N) U^{-1} = H(N)},
$$
i.e. the Clifford group is the normalizer $N_{\U(N)}(H(N))$ of $H(N)$  in $\U(N)$.
\footnote{Let us note that in the references \cite{appleby,BengtssonZyczkowski} the Heisenberg groups are defined separately for odd and even $N$ and they are, moreover, finite groups. Nevertheless, the normalizer of these groups in $\U(N)$, i.e. the Clifford groups $C(N)$, is the same as in our case.}
The Clifford group consists therefore of all $U \in \U(N)$ such that
$UW(i,j) U^{-1} \in H(N)$ for all $i,j=0,1,\dots,N-1$.
Now the subgroup $H(N)$ is a normal subgroup of the normalizer,
so one obtains a short exact sequence of group homomorphisms
\begin{equation} \label{8}
     1 \rightarrow H(N) \rightarrow C(N)
\rightarrow C(N)/H(N)  \rightarrow 1 .
\end{equation}
The full structure of the normalizer
$C(N) = N_{\U(N)}(H(N))$ is complicated by phase factors
and rather difficult to describe in general \cite{Hostens}.

Since the Clifford operations, i.e. the corresponding inner automorphisms, 
are not affected by phase factors, it is natural to study the 
\emph{projective Clifford group} as the quotient
$ \overline{C}(N) =  C(N)/ I(N)$.

There is yet another way how to deal with this cosets of $\sim$. If we consider for $U\in \U(N)$  the $\Ad$-action 
$ \Ad_U : \ Y \mapsto UYU^{-1}$, where $Y\in \GL(N,\C)$, then it is easy to verify that $\Ad_U =\Ad_V $ holds if and only if $U\sim V$.

Using this simple fact we can use results in  \cite{HPPT02, KorbTolar10, KorbTolar12, TolarJPCS18} to describe the factor $C(N)/H(N)\cong \overline{C}(N)/\overline{H}(N)$. 
%As we will deal mainly with operators up the scalar multiple, we still define simplified Weyl operators in the form $\widetilde{W}(i,j)=X_N^iZ_N^j$ for $(i,j)\in\Z_N^2$.

% Concerning the quotient group in (\ref{9}) we refer to the theorem
% proved in \cite{HPPT02}:
\begin{theorem}\cite{HPPT02}\label{htpp}
Let $N\geq 2$ be an integer. For every $U_{/\sim}\in \overline{C}(N)$ there is a unique matrix $A\in \SL(2,\mathbb{Z}_N)$ such that $UW(u)U^{-1}\sim W(Au)$ for every $u\in\Z_N^2$. Moreover, the map $$\overline{\pi}: \overline{C}(N)\to \SL(2,\mathbb{Z}_N), \ \ U_{/\sim}\mapsto A$$ is a group epimorphism and $\ker \overline{\pi}=\overline{H}(N)$. %Hence $\overline{C}(N)/\overline{H}(N)\cong \SL(2,\mathbb{Z}_N)$.
It follows that $$\pi: C(N)\to \SL(2,\mathbb{Z}_N), \ \ U\mapsto \overline{\pi}(U_{/\sim})$$ is a group epimorphism as well, $\ker \pi=H(N)$ and $$C(N)/H(N)\cong\overline{C}(N)/\overline{H}(N)\cong \SL(2,\mathbb{Z}_N)\ .    $$
\end{theorem}
% Thus the groups entering (\ref{9}) are isomorphic to
% \begin{eqnarray}
% \P_N  \cong   \mathbb{Z}_N \times \mathbb{Z}_N \equiv \Z_N^2, \\
% \bar{C}(N) / \P_N  \cong 
% N_{\mathcal{M}_N}(\P_N)/\P_{N}   \cong  \SL(2,\mathbb{Z}_N).
%\end{eqnarray}
% Here $\P_N$ is a
%normal subgroup of the normalizer $N_{\mathcal{M_N}}(\P_N)$ with two
%commuting generators $\Ad_{Z_N}$, $\Ad_{X_N}$.
We can therefore write the following exact sequences
\begin{equation} \label{exact1}
     1 \rightarrow \Z_N^2 \stackrel{\nu}{\longrightarrow} \overline{C}(N)
\stackrel{\overline{\pi}}{\longrightarrow} \SL(2,\mathbb{Z}_N)  \rightarrow 1 .
\end{equation}
where $\nu(i,j)=W(i,j)_{/\sim}$ for every $(i,j)\in\Z^2_N$ and

\begin{equation} \label{exact2}
     1 \rightarrow H(N) \stackrel{\iota}{\longrightarrow} C(N)
\stackrel{\pi}{\longrightarrow} \SL(2,\mathbb{Z}_N)  \rightarrow 1 .
\end{equation}
where $\iota$ denotes the inclusion.

\section{Semidirect product and right splitting exact sequence}\label{semidirect}

Our main aim is to investigate the (projective) Clifford group from the view of the semidirect product. Let us therefore briefly recall the basic properties of semidirect product: Let $H$ and $K$ be groups and $\varphi: K\to Aut(H)$ be a group homomorphism from $K$ to the group $Aut(H)$ of all automorphisms of $H$. The outer semidirect product $K\ltimes_{\varphi}H$ is defined as a group with the underlying set $K\times H$ and the multiplication is given by the formula
 $$(x,h)\cdot (y,g)=(xy, h\cdot\varphi(x)(g))$$
 for all $(x,h),(y,g)\in K\times H$.
 
 For an exact sequence of homomorphisms of groups 
$$ 
\begin{tikzcd}
1 \arrow{r}{} &H \arrow{r}{i} & G \arrow{r}{\pi}  & K  \arrow{r}{} &1
% 1 \arrow[r,""] &H \arrow[r,"i"] & G \arrow[r,"\pi"]  & K  \arrow[r,""] &1
\end{tikzcd}
$$

the following are equivalent:
\begin{enumerate}
 \item[(i)] There is a group homomorphism $\varphi: K\to Aut(H)$ and a group isomorphism $\Phi: K\ltimes_{\varphi} H \to G$ such that $\Phi(1,h)=i(h)$ and $\pi(\Phi(x,1))=x$ for every $h\in H$ and $x\in K$. 
 \item[(ii)] The sequence is right splitting, i.e. there is a homomorphism $\mu: K\to G$ such that $\pi\circ\mu=id_K$.
\end{enumerate}

In particular, $\varphi$ can be expressed by means of $\mu$ in the form $$\varphi(x)(h)=i^{-1}(\mu(x)\cdot i(h)\cdot \mu(x)^{-1})$$ for all $x\in K$ and $h\in H$.

It is well known \cite{DuttaPrasad,Gross} that for $N$ odd both the Clifford group and the projective Clifford group are semidirect products arising from the exact sequences (\ref{exact1}) and (\ref{exact2})
$$
\overline{C}(N) \cong   \SL(2,\mathbb{Z}_N)\ltimes_{\overline{\varphi}} \mathbb{Z}_N^2 ,
$$
where $\overline{\varphi}(A)(u)=Au$ for $A\in \SL(2,\mathbb{Z}_N)$ and $u\in\Z_{N}^2$
 and
$$
C(N) \cong   \SL(2,\mathbb{Z}_N)\ltimes_{\varphi} H(N) ,
$$
where $\varphi(A)(\lambda W(u))=\lambda W(Au)$ for $A\in \SL(2,\mathbb{Z}_N)$, $\lambda\in \U(1)$ and $u\in\Z_{N}^2$.

%  Let $\tau=-e^{\frac{\mathbf{i}\pi}{N}}$. Then the multiplicative order of $\tau$ is equal to $N$ if $N$ is odd and equal to $2N$ if $N$ is even.
% The Weyl operators are defined according to \cite{appleby}, as $W(i,j)=\tau^{ij}X_N^iZ_N^j$ for $(i,j)\in\Z_N^2$ if $N$ is odd and $(i,j)\in \Z_{2N}^2$ if $N$ is even. 

This fact constitutes the motivation
to explore in detail the case of the semidirect product 
for even $N$ in the present article. 

%But even dimensions, to which this paper is devoted, are considerably more difficult to handle.

% They consist of  $N^3$  elements for odd $N$,
% $$  H(N) = \Pi_N = \left\lbrace \omega_N^l Z^i_N X_N^j  \ |
% \   l,i,j = 0,1,\ldots,N-1\right\rbrace,
% $$
%  and $2N^3$ elements for even $N$,
% $$  H(N) = \left\lbrace \tau_N^l Z^i_N X_N^j \ |
% \     i,j = 0,1,\ldots,N-1, l = 0,1,\ldots,2N-1 \right\rbrace
% \quad \mbox{for even} \  N. $$

\section{Projective Clifford groups in even dimensions}\label{section_6}

Our point of departure is Theorem 1 in Appleby's article
\cite{appleby}.
% According to this theorem for odd $N$ there is an isomorphism
% $$
% C(N)/U(1) \cong SL(2,\Z_{N})\ltimes_{\varphi} \Z_{N}^2,
% $$
%  where $\phi(A)(v)=Av$ for all $A\in $
 For a matrix $A\in \SL(2,\Z_{2N})$ let $[A]_N\in \SL(2,\Z_{N})$ denote such a matrix where all entries in $A$ are considered to be modulo $N$. 
 In general, the relation modulo $k\in\N$ will be denoted  by $\equiv_k$.

%Notice that in case of $N$ even the cosets $W(i,j)_{/\sim}$ are well and uniquely defined by integers $i,j$ modulo $N$. 
 
\begin{theorem}\cite[Theorem 1]{appleby}\label{appleby}
 Let $N$ be an even integer and let $$\psi:\SL(2,\Z_{2N})\to Aut(\Z_N^2)$$ be a homomorphism defined as $\psi(A)(u)=[A]_N\cdot u$ for all $A\in \SL(2,\Z_{2N})$ and $u\in\Z_N^2$.
 Then there is a group epimorphism 
 $$f: \SL(2,\Z_{2N})\ltimes_\psi \Z_{N}^2 \to \overline{C}(N), \ \ (A,u)\mapsto U_{/\sim}$$
 with the property $UW(w)U^{-1}\sim W([A]_N\cdot w)$ for every $w\in\Z_{N}^2$ and with the kernel
 $$
K=\ker f=\left\{\left(\mat{1+Nr}{Ns}{Nt}{1+Nr}, \left(\begin{array}{c}\frac{N}{2}s\\[2mm]\frac{N}{2}t\end{array}\right) 
\right)\Big| r,s,t\in\{0,1\}\right\}
$$ 
consisting of $8$ elements. Hence the projective Clifford group is described by the isomorphism   $$\widetilde{f}: \big(\SL(2,\Z_{2N})\ltimes_\psi \Z_{N}^2\big)/K \to \overline{C}(N), \ \ (A,u)_{/K}\mapsto f(A,u)\ .$$
\end{theorem}

Using this theorem we can state  the following diagram 

\begin{equation} 
\begin{tikzcd}\label{diagram}
\overline{C}(N) \arrow{r}{\overline{\pi}}  & \SL(2,\Z_{N}) \\%
\big(\SL(2,\Z_{2N})\ltimes_\psi \Z_{N}^2\big)/K \arrow[swap]{ur}{\sigma=\overline{\pi}\circ\widetilde{f}}  \arrow{u}{\widetilde{f}} & 
\end{tikzcd}
\end{equation}
where  $\sigma\left((A,u)_{/K}\right)=[A]_N$ for every $(A,u)\in \SL(2,\Z_{2N})\ltimes \Z_{N}^2$.
\footnote{Indeed, by the definition of $\sigma$ we have $\sigma\left((A,u)_{/K}\right)=\overline{\pi}(U_{/\sim})$ for some $U\in C(N)$ such that $U_{/\sim}=f(A,u)$. By Theorem \ref{appleby}, it holds that $UW(w)U^{-1}\sim W([A]_N\cdot w)$ for every $w\in\Z_{N}^2$. Hence the matrix $[A]_N$ has the property defining the map $\overline{\pi}$ and therefore,  by Theorem \ref{htpp}, we obtain the equality $\sigma\left((A,u)_{/K}\right)=\overline{\pi}(U_{/\sim})=[A]_N$.}

This observation allows us to characterize the property that the projective Clifford groups is a (natural) semidirect product.

\begin{proposition}\label{character_1}
 Let $N$ be an even integer. Then the projective Clifford group $\overline{C}(N)$ is a semidirect product corresponding to the exact sequence $(\ref{exact1})$ (in terms of Section $\ref{semidirect}$) if and only if there is a group homomorphism $$\mu:  \SL(2,\mathbb{Z}_N) \to \big(\SL(2,\Z_{2N})\ltimes_\psi \Z_{N}^2\big)/K$$
 such that  $\sigma\circ\mu=id_{\SL(2,\Z_N)}$.
\end{proposition}
\begin{proof}
 According to Section \ref{semidirect}, the group $\overline{C}(N)$ is a semidirect product corresponding to the exact sequence $(\ref{exact1})$
  if and only if the sequence $(\ref{exact1})$ is right splitting. This means that there is a homomorphism $\mu':\SL(2,\mathbb{Z}_N)\to \overline{C}(N)$ such that $\overline{\pi}\circ\mu'=id_{\SL(2,\Z_N)}$. Now,  $\mu=\widetilde{f}^{-1}\circ \mu'$ has (according to (\ref{diagram})) the desired property and, conversely, having the homomorphism $\mu$ we simply set $\mu'=\widetilde{f}\circ \mu$.
\end{proof}

Notice that the keypoint in Proposition \ref{character_1} is that $\mu$ has to be a \emph{group  homomorphism} (constructing $\mu$ only as a map between sets is no problem). For the purpose of constructing of homomorphisms between groups suitable presentations of groups are a very useful tool. 

Let us briefly recall the main idea. Let $K$ be a group with generators $e_1,\dots,e_n\in K$ that fulfill the defining relation of $K$ presented by some set  $\set{w_{\alpha}(x_1,\dots,x_n)}{\alpha\in S}$ of terms  in free variables $x_1,\dots,x_n$. This means that $w_{\alpha}(e_1,\dots,e_n)=1$ for every $\alpha\in S$ and the group $K$ is isomorphic to a factor of a free group with the free basis $\{x_1,\dots,x_n\}$, where the corresponding kernel  is determined by the set $\set{w_\alpha}{\alpha\in S}$. Assume now, that for a group $G$ there are elements $g_1,\dots,g_n\in G$ such that $w_{\alpha}(g_1,\dots,g_n)=1$ for every $\alpha\in S$. Then there is a unique group homomorphism $\mu:K\to G$ such that $\mu(e_i)=g_i$ for every $i=1,\dots,n$.

We will use the presentation of $\SL(2,\Z_N)$ with  two generators $t$ and $r$ fulfilling the defining relations:
\begin{itemize}
	\item $t^{N}=1$, $r^N=1$,
	\item $t^{k}r^{\ell}=r^{\ell}t^{k}$ for $k, \ell\in\N$ such that $N=k\ell$ and $\gcd(k,\ell)=1$,
	\item $(t^k r^{\ell}t^k)^2=(trt)^2$ for $1\leq k,\ell\leq N-1$ such that $k\ell\equiv_N 1$,
	\item $t^k r^{\ell}t^k=r^{\ell}t^k r^{\ell}$ for $1\leq k,\ell\leq N-1$ such that $k\ell\equiv_N 1$.
\end{itemize}
The corresponding matrices in $\SL(2,\Z_N)$ are $t=\left[\begin{smallmatrix}1 & 1\\ 0 & 1\end{smallmatrix}\right]_{N}$ and $r=\left[\begin{smallmatrix}1 & 0\\ -1 & 1\end{smallmatrix}\right]_{N}$. 

This presentation is equivalent to the presentation introduced in \cite[Chapter 1, Section 1.4]{hopf}, where we have added the condition $r^N=1$ to allow only the non-negative powers in the remaining defining conditions.

% Our goal is to show that there is no monomorphism 
% $$
% \nu:SL(2,\Z_N)\to \left(SL(2,\Z_{2N})\ltimes \Z_{N}^2\right)/K
% $$ 
% such that for 
% $$\pi:\big(SL(2,\Z_{2N})\ltimes \Z_{N}^2\big)/K\to SL(2,\Z_N)$$
% $$(M,h)_{/K}\mapsto [M]_{N}$$
%  $\pi\circ\nu=id_{SL(2,\Z_N)}$ holds. 
%  This means that the short exact sequence  (\ref{9})  of projective groups 
% $$\underbrace{H(N)/U(1)}_{\cong\Z_{N}^2}\to C(N)/U(1)\to 
% \underbrace{C(N)/H(N)}_{\cong SL(2,\Z_N)}$$ 
%  is non-splitting, or equivalently, that  the projective Clifford group 
%  $C(N)/U(1)$ is not a natural semidirect product of $SL(2,\Z_N)$ and $\Z_{N}^2$.

Summarizing our approach we obtain the following characterization.

\begin{proposition}\label{character_2}
  Let $N$ be an even integer. Then the projective Clifford group $\overline{C}(N)$ is a semidirect product corresponding to the exact sequence $(\ref{exact1})$ (in terms of Section $\ref{semidirect}$) if and only if there are $T,R\in \SL(2,\Z_{2N})\ltimes_{\psi} \Z_{N}^2$ fulfilling the conditions:
\begin{enumerate}
\item[(i)] $\sigma(T_{/K})=t=\left[\begin{smallmatrix}1 & 1\\ 0 & 1\end{smallmatrix}\right]_{N}$ and $\sigma(R_{/K})=r=\left[\begin{smallmatrix}1 & 0 \\-1 & 1\end{smallmatrix}\right]_{N}$,
 \item[(ii)] $T^{N},R^N\in K$
 \item[(iii)] $T^{k}R^{\ell}(R^{\ell}T^{k})^{-1}\in K$ for $k, \ell\in\N$ such that $N=k\ell$ and $\gcd(k,\ell)=1$
 \item[(iv)] $(T^k R^{\ell}T^k)^2(TRT)^{-2}\in K$ for $1\leq k,\ell\leq N-1$ such that $k\ell\equiv_{N} 1$
 \item[(v)] $T^k R^{\ell}T^k(R^{\ell}T^k R^{\ell})^{-1}\in K$ for $1\leq k,\ell\leq N-1$ such that $k\ell\equiv_{N} 1$.
\end{enumerate}
In such a case we set the homomorphism $\mu$ in Proposition $\ref{character_1}$ as (a unique) extension of the map $t\mapsto T_{/K}$, $r\mapsto R_{/K}$.
\end{proposition}
\begin{proof}
 By Proposition \ref{character_1}, the projective Clifford group $\overline{C}(N)$ is a semidirect product corresponding to the exact sequence $(\ref{exact1})$ if and only if there is a group homomorphism $\mu:  \SL(2,\mathbb{Z}_N) \to \big(\SL(2,\Z_{2N})\ltimes_\psi \Z_{N}^2\big)/K$
 such that  $\sigma\circ\mu=id_{\SL(2,\Z_N)}$ holds. 
 
 Having such a $\mu$, there are $T,R\in \SL(2,\Z_{2N})\ltimes_{\psi} \Z_{N}^2$ such that $\mu(t)=T_{/K}$ and $\mu(r)=R_{/K}$. Plugging these elements into the defining relations of $\SL(2,\mathbb{Z}_N)$ we easily obtain that the corresponding terms belong to the group $K$. 
 
 Conversely, assuming that the  conditions (i)-(v) are fulfilled for some $T,R\in \SL(2,\Z_{2N})\ltimes_{\psi} \Z_{N}^2$, there is a unique homomorphism $\mu:\SL(2,\mathbb{Z}_N) \to \big(\SL(2,\Z_{2N})\ltimes_\psi \Z_{N}^2\big)/K$ such that $\mu(t)=T_{/K}$ and $\mu(r)=R_{/K}$. Since $\sigma(\mu (t))=t$, $\sigma(\mu (r))=r$ and $t$ and $r$ are the generators of $\SL(2,\Z_{N})$, it follows that $\sigma\circ\mu=id_{\SL(2,\Z_{N})}$.
\end{proof}

% To procede with the proof by contradiction 
% we will assume the converse, i.e. that such a monomorphism $\nu$ exists. 

In the rest of this paper we are going to show that for an even integer $N$ the conditions (i)--(v) in Proposition \ref{character_2} lead to a contradiction if and only if $N$ is divisible by $4$ (Theorems \ref{6.1} and \ref{6.4}). 
%In this way we shall arrive at conditions under which $\nu$ doesn't exist.

\section{Preliminary computations}\label{section_7}

Let $N$ be an even number (in this and all remaining sections). To simplify our notation we will write  $\SL(2,\Z_{2N})\ltimes \Z_{N}^2$ instead of $\SL(2,\Z_{2N})\ltimes_{\psi} \Z_{N}^2$.
For the computational purposes, we define the following  maps $\g$ and $\f$ that determine both components of an element $(C,w)$ of the respective semidirect group product
$$\g:\SL(2,\Z_{2N})\ltimes \Z_{N}^2\to \SL(2,\Z_{2N}), \ \ (C,w)\mapsto C$$ and
$$\f:\SL(2,\Z_{2N})\ltimes \Z_{N}^2\to \Z_{N}^2, \ (C,w)\mapsto w\ .$$ 
The important basic properties of these maps are summarized in the following proposition.
\begin{proposition}\label{basic_0}
Let $S,P\in \SL(2,\Z_{2N})\ltimes \Z_{N}^2$. Then:
\begin{enumerate}
 \item[(i)] $\g(S\cdot P)=\g(S)\cdot\g(P)$, i.e., $\g$ is a group homomorphism.
 \item[(ii)] $\f(S\cdot P)=\f(S)+[\g(S)]_{N}\cdot\f(P)$.
 \item[(iii)] If $[\g(S)]_N=[\g(P)]_N$ then $\f(SP^{-1})=\f(S)-\f(P)$. 
\end{enumerate}
\end{proposition}
\begin{proof}
 (i) and (ii) follow immediately from the definition of the semidirect product $\SL(2,\Z_{2N})\ltimes \Z_{N}^2$.
 
 (iii) It is easy to check that $P^{-1}=\big(\g(P)^{-1},-[\g(P)^{-1}]_N\cdot \f(P)\big)$. Assume now that $[\g(S)]_N=[\g(P)]_N$. Then $[\g(S)]_N\cdot [\g(P)^{-1}]_N=[I]_N$ and, by the part (ii), we obtain  $$\f(SP^{-1})=\f(S)+ [\g(S)]_N\cdot \f(P^{-1})=\f(S)- [\g(S)]_N\cdot [\g(P)^{-1}]_N\cdot \f(P)=$$
 $$=\f(S)- [I]_N\cdot \f(P)=\f(S)- \f(P)\ .$$
\end{proof}

%These important properties will be used in the sections that follow. 

Through the rest of the paper let $$T=(A,h)\in \SL(2,\Z_{2N})\ltimes \Z_{N}^2\ \ \text{and} \  \
 R=(B,h')\in \SL(2,\Z_{2N})\ltimes \Z_{N}^2$$ fulfill the condition (i) in Proposition \ref{character_2}, i.e. 
$$\sigma(T/_{K})=\matN{1}{1}{0}{1}\ \ \ \text{and} \ \ \ \sigma(R/_{K})=\matN{1}{0}{-1}{1} .
$$
By the definition of $\sigma$ (see the diagram (\ref{diagram})) we have
$$\matN{1}{1}{0}{1} =\sigma(T/_{K})=\sigma\left((A,h)_{/K}\right)=[A]_{N}$$
and $$\matN{1}{0}{-1}{1} =\sigma(R/_{K})=\sigma\left((B,h')_{/K}\right)=[B]_{N}\ .$$
The kernel of the map $\SL(2,\Z_{2N})\to \SL(2,\Z_{N})$, $E\mapsto[E]_N$ consists of elements $\left(\begin{smallmatrix}1+\alpha N & \beta N\\ \gamma N & 1+\alpha N\end{smallmatrix}\right)\in \SL(2,\Z_{2N})$ where $\alpha,\beta,\gamma\in\{0,1\}$. It follows therefore that
$$A=\mat{1}{1}{0}{1}\cdot\mat{1+aN}{bN}{cN}{1+aN}=\mat{1}{1}{0}{1}+N\mat{a+c}{a+b}{c}{a}$$
and 
$$B=\mat{1}{0}{-1}{1}\cdot\mat{1+a'N}{b'N}{c'N}{1+a'N}=\mat{1}{0}{-1}{1}+N\mat{a'}{b'}{c'-a'}{a'-b'}$$
for some $a,b,c,a',b',c'\in\{0,1\}$.

In the rest of the paper  we fix the meaning of $T,R, A, B$ and $a,b,c,a',b',c'$ and $h=\left(\begin{smallmatrix}u\\v\end{smallmatrix}\right)\in\Z_N^2$ and $h'=\left(\begin{smallmatrix}u'\\v'\end{smallmatrix}\right)\in\Z_N^2$.

Recall that $$\g(T)=A, \ \ \ \ \g(S)=B $$
and $$\f(T)=h, \ \ \ \ \f(S)=h' \ .$$

%The elements $\tilde{t}, \tilde{r}\in \big(SL(2,\Z_{2N})\ltimes \Z_{N}^2\big)/K$ fullfil the generating relations in $SL(2,\Z_{2N})\ltimes \Z_{N}^2$.

\

%First we will do some preliminary computations.

% \begin{proof}
%  ${N-k\choose 2}=\frac{(N-k)(N-k-1)}{2}=\frac{N^2-2kN+k^2-N+k}{2}\equiv_N\frac{k^2+k-N}{2}={k+1\choose 2}-\frac{N}{2}$
% \end{proof}

\begin{lemma}\label{basic}
Let $k\in\N$. Then:
\begin{enumerate}
\item[(i)] $(C,w)^k=\left(C^k,\left[\sum_{i=0}^{k-1}C^i\right]_N\cdot w\right)$ for $(C,w)\in \SL(2,\Z_{2N})\ltimes \Z_{N}^2$.
\item[(ii)] Let $U$ and $G$ be square matrices over $\Z_{2N}$. Then 
$$(U+N G)^k=U^k+N\sum\limits^{k}_{i=1}U^{i-1}GU^{k-i}\ .$$
\end{enumerate}
\end{lemma}
\begin{proof}
 (i) We proceed by induction:  
 \begin{align*}
 (C,w)^{k+1}&=(C,w)\cdot(C,w)^k=(C,w)\cdot\left(C^k,\left[\sum_{i=0}^{k-1}C^i\right]_N\cdot w\right)=\\
 &=\left(C^{k+1},w+[C]_N\cdot\left[\sum_{i=0}^{k-1}C^i\right]_N\cdot w\right)=\left(C^{k+1},\left[\sum_{i=0}^{k}C^i\right]_N\cdot w\right)\ .
 \end{align*}

 (ii) Follows easily by induction and the fact that $N^{2}\equiv_{2N}0$.
\end{proof}

\begin{remark}\label{7.2}
 For $k\in\N$ we use the following identities  $$\sum_{i=1}^k i={k+1\choose 2}, \ \ \ \ \sum_{i=1}^k (k-i)={k\choose 2}, \ \ \ \ \sum_{i=1}^k i(k-i)={k+1\choose 3}\ .$$
\end{remark}

\begin{lemma}\label{product}
Let $k,\ell\in\N$. Then:
\begin{enumerate}
 \item[(i)] $\f(T^k)=\matN{k}{{k\choose 2}}{0}{k}\cdot h$ \ \  and \ \ $\f(R^\ell)=\matN{\ell}{0}{-{\ell\choose 2}}{\ell}\cdot h'$,
 \item[(ii)] $\f(T^k\cdot R^\ell)= \matN{k}{{k\choose 2}}{0}{k}\cdot h+\matN{\ell-k{\ell\choose 2}}{k\ell}{-{\ell\choose 2}}{\ell}\cdot h'$,
 \item[(iii)] $\f(R^\ell\cdot T^k)=\matN{k}{{k\choose 2}}{-k\ell}{k-\ell{k\choose 2}}\cdot h+\matN{\ell}{0}{-{\ell\choose 2}}{\ell}\cdot h'$.
\end{enumerate}
\end{lemma}
\begin{proof}
(i) By Lemma \ref{basic}(i) we have the equality $\f(T^k)=\left[\sum_{i=0}^{k-1}A^i\right]_N\cdot h$ and $\f(R^k)=\left[\sum_{i=0}^{k-1}B^i\right]_N\cdot h'$. Using  Remark \ref{7.2}, we obtain $$\left[\sum_{i=0}^{k-1}A^i\right]_N=\sum_{i=0}^{k-1}\left[A\right]^i_N=\sum_{i=0}^{k-1}\matN{1}{1}{0}{1}^i=\sum_{i=0}^{k-1}\matN{1}{i}{0}{1}=\matN{k}{{k\choose 2}}{0}{k}$$ and
 $$\left[\sum_{i=0}^{\ell-1}B^i\right]_N=\sum_{i=0}^{\ell-1}\left[B\right]^i_N=\sum_{i=0}^{\ell-1}\matN{1}{0}{-1}{1}^i=\sum_{i=0}^{\ell-1}\matN{1}{0}{-i}{1}=\matN{\ell}{0}{-{\ell\choose 2}}{\ell} \ .$$
The assertion now immediateky follows.
 
(ii) and (iii).  By the previous part (i), we obtain
\begin{align*}
\f(T^k\cdot R^\ell)&=\f(T^k)+[\g(T^k)]_{N}\cdot\f(R^\ell)=\\
&=\matN{k}{{k\choose 2}}{0}{k}\cdot h+\left[A^k\right]_N\cdot\matN{\ell}{0}{-{\ell\choose 2}}{\ell}\cdot h'=\\
&=\matN{k}{{k\choose 2}}{0}{k}\cdot h+\matN{1}{k}{0}{1}\cdot\matN{\ell}{0}{-{\ell\choose 2}}{\ell}\cdot h'=\\
&=\matN{k}{{k\choose 2}}{0}{k}\cdot h+\matN{\ell-k{\ell\choose 2}}{k\ell}{-{\ell\choose 2}}{\ell}\cdot h'
\end{align*}
and
\begin{align*}
\f(R^\ell\cdot T^k)&=\f(R^{\ell})+[\g(R^{\ell})]_{N}\cdot\f(T^k)=\\
&=\matN{\ell}{0}{-{\ell\choose 2}}{\ell}\cdot h'+\left[B^\ell\right]_N\cdot\matN{k}{{k\choose 2}}{0}{k}\cdot h=\\
&=\matN{\ell}{0}{-{\ell\choose 2}}{\ell}\cdot h'+\matN{1}{0}{-\ell}{1}\cdot\matN{k}{{k\choose 2}}{0}{k}\cdot h=\\
&=\matN{\ell}{0}{-{\ell\choose 2}}{\ell}\cdot h'+\matN{k}{{k\choose 2}}{-k\ell}{k-\ell{k\choose 2}}\cdot h\ .
\end{align*}
\end{proof}

\begin{lemma}\label{7.4}
Let $k,\ell\in\N$. Then:
\begin{enumerate}
\item[(i)] $A^k=\mat{1}{k}{0}{1}+N\mat{ka+{k+1\choose 2}c}{k^2a+kb+{k+1\choose 3}c}{kc}{ka+{k\choose 2}c}$,

\item[(ii)] $B^\ell=\mat{1}{0}{-\ell}{1}+N\mat{\ell a'-{\ell\choose 2}b'}{\ell b'}{-\ell^2a'+{\ell+1\choose 3}b'+\ell c'}{\ell a'-{\ell+1\choose 2}b'}\ .$
\end{enumerate}
\end{lemma}
\begin{proof}
By Lemma \ref{basic}(ii) and Remark \ref{7.2}, we obtain
\begin{eqnarray*}
A^k&=&\mat{1}{1}{0}{1}^k+N\sum\limits^{k}_{i=1}\mat{1}{1}{0}{1}^{i-1}\mat{a+c}{a+b}{c}{a}\mat{1}{1}{0}{1}^{k-i}=\\
&=&\mat{1}{k}{0}{1}+N\sum\limits^{k}_{i=1}\mat{1}{i-1}{0}{1}\mat{a+c}{a+b}{c}{a}\mat{1}{k-i}{0}{1}=\\
&=&\mat{1}{k}{0}{1}+N\sum\limits^{k}_{i=1}\mat{a+ic}{ka+b+i(k-i)c}{c}{a+(k-i)c}=\\
&=&\mat{1}{k}{0}{1}+N\mat{ka+{k+1\choose 2}c}{k^2a+kb+{k+1\choose 3}c}{kc}{ka+{k\choose 2}c}
\end{eqnarray*}
and similarly
\begin{eqnarray*}
B^\ell&=&\mat{1}{0}{-1}{1}^\ell+N\sum\limits^{\ell}_{i=1}\mat{1}{0}{-1}{1}^{i-1}\mat{a'}{b'}{c'-a'}{a'-b'}\mat{1}{0}{-1}{1}^{\ell-i}=\\
&=&\mat{1}{0}{-\ell}{1}+N\sum\limits^{\ell}_{i=1}\mat{1}{0}{1-i}{1}\mat{a'}{b'}{c'-a'}{a'-b'}\mat{1}{0}{i-\ell}{1}=\\
&=&\mat{1}{0}{-\ell}{1}+N\sum\limits^{\ell}_{i=1}\mat{a'-(\ell-i)b'}{b'}{-\ell a'+i(\ell-i)b'+c'}{a'-ib'}=\\
&=&\mat{1}{0}{-\ell}{1}+N\mat{\ell a'-{\ell \choose 2}b'}{\ell b'}{-\ell^2a'+{\ell+1\choose 3}b'+\ell c'}{\ell a'-{\ell+1\choose 2}b'}\ .
\end{eqnarray*}
\end{proof}

\begin{remark}\label{7.5}
For $n\in\N_0$ the following holds:
 $${n\choose 2}\equiv_2 
 \begin{cases}
 \frac{n}{2}, & n \ \text{is even,}\\[2mm]
\frac{n-1}{2}, & n \ \text{is odd.}
\end{cases}
\ \ \ \  \ \ \ \ \
{n\choose 3}\equiv_2 
\begin{cases}
0, & n \ \text{is even,}\\[2mm]
\frac{n-1}{2}, & n \ \text{is odd.}
\end{cases}                     
$$
%$${N-k\choose 2}\equiv_N{k+1\choose 2}-\frac{N}{2}$$
\end{remark}

\begin{corollary}\label{7.6}
Let $k,\ell\in\N$. Then:
$$A^k=\begin{cases}
\mat{1}{k}{0}{1}+N\mat{\frac{k}{2}c}{\frac{k}{2}c}{0}{\frac{k}{2}c}, &\ k\ \text{is even,}\\[7mm]
\mat{1}{k}{0}{1}+N\mat{a+\frac{k+1}{2}c}{a+b}{c}{a+\frac{k-1}{2}c}, &\ k\ \text{is odd,}\\
\end{cases}$$
$$B^\ell=\begin{cases}
\mat{1}{0}{-\ell}{1}+N\mat{\frac{\ell}{2}b'}{0}{\frac{\ell}{2}b'}{\frac{\ell}{2}b'}, &\ \ell\ \text{is even,}\\[7mm]
\mat{1}{0}{-\ell}{1}+N\mat{a'+\frac{\ell-1}{2}b'}{b'}{a'+c'}{a'+\frac{\ell+1}{2}b'}, &\ \ell\ \text{is odd.}\\
\end{cases}$$
\end{corollary}
\begin{proof}
 Follows immediately from Lemma \ref{7.4} and Remark \ref{7.5} and the fact that for all $j,m\in\N$ is $Nj\equiv_{2N}0$ if $j$ is even, $Nj\equiv_{2N}1$ if $j$ is odd, and $-Nm\equiv_{2N} Nm$.
\end{proof}

In the following Sections \ref{section_8} --\ref{section_11} we express the conditions (ii)--(v) from Proposition \ref{character_2} in terms of $a,b,c,a',b',c'\in\{0,1\}$ and $u,v,u',v'\in\Z_N$ (see the beginning of Section \ref{section_7}) that correspond to elements $T,R\in\SL(2,\Z_{2N})\ltimes \Z_{N}^2$.

\section{Condition $T^{N},R^N\in K$}\label{section_8}
%\section{Condition $\tilde{t}^N=1$ and $\tilde{r}^N=1$}

\begin{lemma}\label{1}
 $$T^N=(A,h)^N=\left(\mat{1}{0}{0}{1}+N\mat{\frac{N}{2}c}{1+\frac{N}{2}c}{0}{\frac{N}{2}c},\vc{\frac{N}{2}\cdot v}{0}\right)$$
 and
 $$R^N=(B,h')^N=\left(\mat{1}{0}{0}{1}+N\mat{\frac{N}{2}b'}{0}{1+\frac{N}{2}b'}{\frac{N}{2}b'},\vc{0}{\frac{N}{2}\cdot u'}\right)$$
 where $\f(T)=h=\left(\begin{smallmatrix}u\\v\end{smallmatrix}\right)\in\Z_N^2$ and $\f(R)=h'=\left(\begin{smallmatrix}u'\\v'\end{smallmatrix}\right)\in\Z_N^2$.
\end{lemma}
\begin{proof}
We use the fact that $T^{k}=(\g(T^k),\f(T^k))$ and, similarly,  $R^{\ell}=(\g(R^\ell),\f(R^\ell))$.
Since $N$ is even, we obtain, by Corollary \ref{7.6}, that 
$$\g(T^k)=A^N=\mat{1}{N}{0}{1}+N\mat{\frac{N}{2}c}{\frac{N}{2}c}{0}{\frac{N}{2}c}$$ and
$$\g(R^\ell)=B^N=\mat{1}{0}{-N}{1}+N\mat{\frac{N}{2}b'}{0}{\frac{N}{2}b'}{\frac{N}{2}b'}=\mat{1}{0}{N}{1}+N\mat{\frac{N}{2}b'}{0}{\frac{N}{2}b'}{\frac{N}{2}b'}$$
where we have used the equality $-N\equiv_{2N}N$.
By Lemma \ref{product}(i),  we then have that
$$\f(T^N)=\matN{N}{{N\choose 2}}{0}{N}\cdot \vc{u}{v}=\vc{{N\choose 2}v}{0}=\vc{\frac{N}{2}v}{0}$$ and
$$\f(R^N)=\matN{N}{0}{-{N\choose 2}}{N}\cdot \vc{u'}{v'}=\vc{0}{-{N\choose 2}u'}=\vc{0}{\frac{N}{2}u'}$$
where we have used that the equality ${N\choose 2}\equiv_N \frac{N}{2}(N-1)\equiv_N -\frac{N}{2}\equiv_N \frac{N}{2}$ holds for the even number $N$.
\end{proof}

\begin{proposition}\label{2.2}
The condition $T^{N},R^N\in K$ holds if and only if  $v\equiv_2 \frac{N}{2}c+1$ and $u'\equiv_2 \frac{N}{2}b'+1$.
\end{proposition}
\begin{proof}
By  Lemma \ref{1} and by the definition of the group $K$, the condition
$$T^N=\left(\mat{1}{0}{0}{1}+N\mat{\frac{N}{2}c}{1+\frac{N}{2}c}{0}{\frac{N}{2}c},\vc{\frac{N}{2}\cdot v}{0}\right)\in K$$ and 
$$R^N=\left(\mat{1}{0}{0}{1}+N\mat{\frac{N}{2}b'}{0}{1+\frac{N}{2}b'}{\frac{N}{2}b'},\vc{0}{\frac{N}{2}\cdot u'}\right)\in K$$
is equivalent to $v\equiv_2 \frac{N}{2}c+1$ and $u'\equiv_2 \frac{N}{2}b'+1$.
\end{proof}

\section{Condition $T^{k}R^{\ell}(R^{\ell}T^{k})^{-1}\in K$ where  $N=k\ell$ and $\gcd(k,\ell)=1$}\label{section_9}

%\section{Condition $\tilde{t}^{k}\tilde{r}^{\ell}=\tilde{r}^{\ell}\tilde{t}^{k}$ where  $N=k\ell$ and $\gcd(k,\ell)=1$}

\begin{lemma}\label{N}
Let $k, \ell\in\N$ be such that $N=k\ell$ and $\gcd(k,\ell)=1$.  Then
$$T^k\cdot R^\ell\left(R^\ell\cdot T^k\right)^{-1}=\left(A^kB^\ell(B^\ell A^k)^{-1},\left(\begin{array}{c}\frac{N}{2}(\ell-1)u'\\[2mm]\frac{N}{2}(k-1)v\end{array}\right)\right)$$
 %\item $$\left((B,h')^\ell\cdot (A,h)^k\right)^{-1}=
%\left((B^\ell A^k)^{-1},-\left[\mat{\ell+\frac{N}{2}(\ell-1)}{0}{{\ell+1\choose 2}}{\ell}\right]_N\cdot h'-\left[\mat{k}{-{k+1\choose 2}}{0}{k}\right]_N\cdot h\right)$$
where $\f(T)=\left(\begin{smallmatrix}u\\v\end{smallmatrix}\right)\in\Z_N^2$ and $\f(R)=\left(\begin{smallmatrix}u'\\v'\end{smallmatrix}\right)\in\Z_N^2$.
\end{lemma}
\begin{proof}
We only need to determine $\f\big(T^k\cdot R^\ell\left(R^\ell\cdot T^k\right)^{-1}\big)$.
Since $k\ell=N$, it holds that $k{\ell\choose 2}=\frac{k\ell(\ell-1)}{2}=\frac{N}{2}(\ell-1)$ and $\ell{k\choose 2}=\frac{\ell k(k-1)}{2}=\frac{N}{2}(k-1)$. Hence, by Lemma \ref{product}(ii) and (iii), we obtain 
$$\f(T^k\cdot R^\ell)=\matN{k}{{k\choose 2}}{0}{k}\cdot h+\matN{\frac{N}{2}(1-\ell)+\ell}{0}{-{\ell\choose 2}}{\ell}\cdot h'$$ and
$$\f(R^\ell\cdot T^k)=\matN{k}{{k\choose 2}}{0}{\frac{N}{2}(1-k)+k}\cdot h+\matN{\ell}{0}{-{\ell\choose 2}}{\ell}\cdot h'\ .$$
By the generating relations for 
$t=\left[\begin{smallmatrix}1 & 1\\ 0 & 1\end{smallmatrix}\right]_{N}=[A]_N$ and $r=\left[\begin{smallmatrix}1 & 0\\ -1 & 1\end{smallmatrix}\right]_{N}=[B]_N$ of the group $\SL(2,\Z_{N})$ it holds that  $$\left[\g(T^k\cdot R^\ell)\right]_{N}=\left[A^kB^\ell\right]_N=t^kr^{\ell}=r^{\ell}t^k=\left[B^\ell A^k\right]_N=\left[\g(R^\ell\cdot T^k)\right]_{N}\ .$$  Then, by Lemma  \ref{basic_0}(iii), it follows that
 $$\f\left(T^k\cdot R^\ell\left(R^\ell\cdot T^k\right)^{-1}\right)=\f\left(T^k\cdot R^\ell\right)-\f\left(R^\ell\cdot T^k\right)=$$
 $$=\matN{0}{0}{0}{\frac{N}{2}(k-1)}\cdot h+\matN{-\frac{N}{2}(\ell-1)}{0}{0}{0}\cdot h'=$$
 $$=\matN{0}{0}{0}{\frac{N}{2}(k-1)}\cdot \vc{u}{v}+\matN{-\frac{N}{2}(\ell-1)}{0}{0}{0}\cdot \vc{u'}{v'}=\left(\begin{array}{c}\frac{N}{2}(\ell-1)u'\\[2mm]\frac{N}{2}(k-1)v\end{array}\right)$$
where we have used the equality $-\frac{N}{2}\equiv_{N}\frac{N}{2}$.
\end{proof}

% \begin{remark}
%  ${k\choose 2}+{N-k\choose 2}-k^2={N\choose 2}$
% \end{remark}

As $N=k\ell$ and $\gcd(k,\ell)=1$ it follows that either $k$ is even and $\ell$ is odd or, vice versa, $k$ is odd and $\ell$ even.

\begin{remark}\label{inverse}
 For $\left(\begin{smallmatrix}\alpha & \beta\\ \gamma & \delta\end{smallmatrix}\right)\in\SL(2,\Z_{2N})$ is $\left(\begin{smallmatrix}\alpha & \beta\\ \gamma & \delta\end{smallmatrix}\right)^{-1}=\left(\begin{smallmatrix}\delta & -\beta\\ -\gamma & \alpha\end{smallmatrix}\right)$.
\end{remark}

\subsection{Case $k$ is  even and $\ell$ is odd}

In this part we assume that $k\in\N$ is even, $\ell\in\N$ is odd,  $N=k\ell$ and $\gcd(k,\ell)=1$.
\begin{lemma}\label{even_odd}
 $A^k B^\ell(B^\ell A^k)^{-1}=\mat{1}{0}{0}{1}+N
 \mat{1+\frac{k}{2}c+a'+c'}{0}{1+\frac{k}{2}c+a'+c'}{1+\frac{k}{2}c}$.
\end{lemma}
\begin{proof}
Since $k$ is even and $\ell$ is odd, it follows that $N\cdot(\begin{smallmatrix}1 & k\\ 0 & 1\end{smallmatrix})=N\cdot(\begin{smallmatrix}1 & 0\\ 0 & 1\end{smallmatrix})$ and $N\cdot(\begin{smallmatrix}1 & 0\\ -\ell & 1\end{smallmatrix})=N\cdot(\begin{smallmatrix}1 & 0\\ 1 & 1\end{smallmatrix})$ when computing modulo $2N$. Using these equalities, the equality $k\ell=N$ and Corollary \ref{7.6} we obtain
\begin{eqnarray*}
 A^k B^\ell &=& \Bigg(\mat{1}{k}{0}{1}+N\tfrac{k}{2}c\mat{1}{1}{0}{1}\Bigg)\cdot
 \\&&\cdot\Bigg(\mat{1}{0}{-\ell}{1}+N\mat{a'+\frac{\ell-1}{2}b'}{b'}{a'+c'}{a'+\frac{\ell+1}{2}b'}\Bigg)=
\\&=&\mat{1-N}{k}{-\ell}{1}+N\tfrac{k}{2}c\mat{1}{1}{0}{1}\mat{1}{0}{1}{1}+\\
&&+N\mat{a'+\frac{\ell-1}{2}b'}{b'}{a'+c'}{a'+\frac{\ell+1}{2}b'}=
\\&=&\mat{1-N}{k}{-\ell}{1}+N\tfrac{k}{2}c\mat{0}{1}{1}{1}+N\mat{a'+\frac{\ell-1}{2}b'}{b'}{a'+c'}{a'+\frac{\ell+1}{2}b'}=
\\&=&\mat{1}{k}{-\ell}{1}+N\mat{1+a'+\frac{\ell-1}{2}b'}{\frac{k}{2}c+b'}{\frac{k}{2}c+a'+c'}{\frac{k}{2}c+a'+\frac{\ell+1}{2}b'}
 \end{eqnarray*}
and, similarly,
 \begin{eqnarray*}
 B^\ell A^k&=&\Bigg(\mat{1}{0}{-\ell}{1}+N\mat{a'+\frac{\ell-1}{2}b'}{b'}{a'+c'}{a'+\frac{\ell+1}{2}b'}\Bigg)\cdot\\
 &&\cdot\Bigg(\mat{1}{k}{0}{1}+N\tfrac{k}{2}c\mat{1}{1}{0}{1}\Bigg)=\\
&=&\mat{1}{k}{-\ell}{1-N}+N\mat{a'+\frac{\ell-1}{2}b'}{b'}{a'+c'}{c'+\frac{\ell+1}{2}b'}+\\
&&+N\tfrac{k}{2}c\mat{1}{0}{1}{1}\mat{1}{1}{0}{1}=\\
&=&\mat{1}{k}{-\ell}{1-N}+N\mat{a'+\frac{\ell-1}{2}b'}{b'}{a'+c'}{c'+\frac{\ell+1}{2}b'}+N\tfrac{k}{2}c\mat{1}{1}{1}{0}=\\
&=&\mat{1}{k}{-\ell}{1}+N\mat{\alpha}{\beta}{\gamma}{\delta}
 \end{eqnarray*}
where $$\mat{\alpha}{\beta}{\gamma}{\delta}=\mat{\frac{k}{2}c+a'+\frac{\ell-1}{2}b'}{\frac{k}{2}c+b'}{\frac{k}{2}c+a'+c'}{1+c'+\frac{\ell+1}{2}b'}\ .$$ 
Hence, by Remark \ref{inverse} and the equality $-N\equiv_{2N}N$, it holds that
$$(B^\ell A^k)^{-1}=\mat{1}{-k}{\ell}{1}+N\mat{\delta}{\beta}{\gamma}{\alpha}$$
and therefore
\begin{align*}
&A^k B^\ell(B^\ell A^k)^{-1}=\Bigg(\mat{1}{k}{-\ell}{1}+N\mat{\delta-b'-c'+a'}{\beta}{\gamma}{\alpha+b'}\Bigg)\cdot\\
&\hspace{60mm}\cdot\Bigg(\mat{1}{-k}{\ell}{1}+N\mat{\delta}{\beta}{\gamma}{\alpha}\Bigg)=\\
&=\mat{1+N}{0}{0}{1+N}+N\mat{1}{0}{1}{1}\mat{\delta}{\beta}{\gamma}{\alpha}+\\
&\hspace{45mm}+N
\mat{\delta-b'-c'+a'}{\beta}{\gamma}{\alpha+b'}\mat{1}{0}{1}{1}=\\
&=\mat{1+N}{0}{0}{1+N}+N
\mat{\delta}{\beta}{\gamma+\delta}{\alpha+\beta}+\\
&\hspace{55mm}+N\mat{\beta+\delta-b'-c'+a'}{\beta}{\alpha+\gamma+b'}{\alpha+b'}=\\
&=\mat{1}{0}{0}{1}+N\mat{1+\beta-b'-c'+a'}{0}{\alpha+\delta+b'}{1+\beta+b'}=\\
&=\mat{1}{0}{0}{1}+N\mat{1+\frac{k}{2}c+a'+c'}{0}{1+\frac{k}{2}c+a'+c'}{1+\frac{k}{2}c}\ .
 \end{align*}
\end{proof}

\begin{proposition}\label{3.3}
The following conditions are equivalent:
\begin{enumerate}
\item[(i)] $T^{k}R^{\ell}(R^{\ell}T^{k})^{-1}\in K$  holds for all $k,\ell\in \N$ such that $k$ is even, $\ell$ is odd, $N=k\ell$ and $\gcd(k,\ell)=1$. 
 \item[(ii)] $a'\equiv_2 c'$ and $v\equiv_2\begin{cases} 1+c & \text{if}\ N\equiv_4 2\\ 1 & \text{if}\ N\equiv_4 0 .\end{cases}$
\end{enumerate}
\end{proposition}
\begin{proof}
Let $k\in \N$ be even and $\ell\in \N$ be odd such that $N=k\ell$ and $\gcd(k,\ell)=1$. By Lemmas \ref{N} and  \ref{even_odd} and the definition of the group $K$, the condition 
$$T^k\cdot R^\ell\left(R^\ell\cdot T^k\right)^{-1}=\left(A^kB^\ell(B^\ell A^k)^{-1},\vc{\frac{N}{2}(\ell-1)u'}{\frac{N}{2}(k-1)v}\right)=$$
$$=\left(\mat{1}{0}{0}{1}+N
 \mat{1+a'+c'+\frac{k}{2}c}{0}{1+a'+c'+\frac{k}{2}c}{1+\frac{k}{2}c},\vc{0}{\frac{N}{2}v}\right)\in K$$
is equivalent to 
$a'\equiv_2 c'$  and 
$v\equiv_2 1+\frac{k}{2}c\ .$
The rest is easy.
\end{proof}

\subsection{Case $k$ is odd and $\ell$ is even}

In this part we assume that $k\in\N$ is odd, $\ell\in\N$ is even, $N=k\ell$ and $\gcd(k,\ell)=1$.

\begin{lemma}\label{odd_even}
 $A^k B^\ell(B^\ell A^k)^{-1}=\mat{1}{0}{0}{1}+N\mat{1+\frac{\ell}{2}b'}{1+\frac{\ell}{2}b'}{0}{1+\frac{\ell}{2}b'}\ .$
\end{lemma}
\begin{proof}
We proceed similarly as in the proof of Lemma \ref{even_odd}. Since $k$ is odd and $\ell$ is even, it follows that $N\cdot(\begin{smallmatrix}1 & k\\ 0 & 1\end{smallmatrix})=N\cdot(\begin{smallmatrix}1 & 1\\ 0 & 1\end{smallmatrix})$ and $N\cdot(\begin{smallmatrix}1 & 0\\ -\ell & 1\end{smallmatrix})=N\cdot(\begin{smallmatrix}1 & 0\\ 0 & 1\end{smallmatrix})$ when computing modulo $2N$. Using these equalities, the equality $k\ell=N$ and Corollary \ref{7.6} we obtain
\begin{align*}
 &A^k B^\ell=\Bigg(\mat{1}{k}{0}{1}+
 N\mat{a+\frac{k+1}{2}c}{a+b}{c}{a+\frac{k-1}{2}c}\Bigg)\cdot\\
 &\hspace{65mm}\cdot\Bigg(\mat{1}{0}{-\ell}{1}+
 N\tfrac{\ell}{2}b'\mat{1}{0}{1}{1}\Bigg)=\\
 &=\mat{1-N}{k}{-\ell}{1}+
N\mat{a+\frac{k+1}{2}c}{a+b}{c}{a+\frac{k-1}{2}c}+
N\tfrac{\ell}{2}b'\mat{1}{1}{0}{1}\mat{1}{0}{1}{1}=\\
&=\mat{1-N}{k}{-\ell}{1}+
N\mat{a+\frac{k+1}{2}c}{a+b}{c}{a+\frac{k-1}{2}c}+
N\tfrac{\ell}{2}b'\mat{0}{1}{1}{1}=\\
&=\mat{1}{k}{-\ell}{1}+
N\mat{1+a+\frac{k+1}{2}c}{a+b+\frac{\ell}{2}b'}{c+\frac{\ell}{2}b'}{a+\frac{k-1}{2}c+\frac{\ell}{2}b'}
\end{align*}
and, similarly,
\begin{align*}
 &B^\ell A^k=\Bigg(\mat{1}{0}{-\ell}{1}+
 N\tfrac{\ell}{2}b'\mat{1}{0}{1}{1}\Bigg)\cdot\\
 &\hspace{50mm}\cdot\Bigg(\mat{1}{k}{0}{1}+
 N\mat{a+\frac{k+1}{2}c}{a+b}{c}{a+\frac{k-1}{2}c}\Bigg)=\\
 &=\mat{1}{k}{-\ell}{1-N}+
 N\mat{a+\frac{k+1}{2}c}{a+b}{c}{a+\frac{k-1}{2}c}+N\tfrac{\ell}{2}b'\mat{1}{0}{1}{1}\mat{1}{1}{0}{1}=\\
&=\mat{1}{k}{-\ell}{1-N}+
 N\mat{a+\frac{k+1}{2}c}{a+b}{c}{a+\frac{k-1}{2}c}+N\tfrac{\ell}{2}b'\mat{1}{1}{1}{0}=\\
&=\mat{1}{k}{-\ell}{1}+N
\mat{\alpha}{\beta}{\gamma}{\delta}
\end{align*}
 where $$\mat{\alpha}{\beta}{\gamma}{\delta}=\mat{a+\frac{k+1}{2}c+\frac{\ell}{2}b'}{a+b+\frac{\ell}{2}b'}{c+\frac{\ell}{2}b'}{1+a+\frac{k-1}{2}c}\ .$$
Hence, by Remark \ref{inverse} and by the equality $-N\equiv_{2N} N$, it holds that 
$$(B^\ell A^k)^{-1}=\mat{1}{-k}{\ell}{1}+
N\mat{\delta}{\beta}{\gamma}{\alpha}$$
and therefore
\begin{align*}
&A^k B^\ell(B^\ell A^k)^{-1}=\Bigg(\mat{1}{k}{-\ell}{1}+
N\mat{\delta+c}{\beta}{\gamma}{\alpha-c}\Bigg)\cdot\\
&\hspace{60mm}\cdot\Bigg(\mat{1}{-k}{\ell}{1}+
N\mat{\delta}{\beta}{\gamma}{\alpha}\Bigg)=\\
&=\mat{1+N}{0}{0}{1+N}+N\mat{1}{1}{0}{1}\mat{\delta}{\beta}{\gamma}{\alpha}+\\
&\hspace{60mm}+N\mat{\delta+c}{\beta}{\gamma}{\alpha-c}\mat{1}{1}{0}{1}=\\
&=\mat{1+N}{0}{0}{1+N}+N\mat{\gamma+\delta}{\alpha+\beta}{\gamma}{\alpha}+N\mat{\delta+c}{\beta+\delta+c}{\gamma}{\alpha+\gamma-c}=\\
&=\mat{1}{0}{0}{1}+N\mat{1+\gamma+c}{\alpha+\delta+c}{0}{1+\gamma-c}=\\
&=\mat{1}{0}{0}{1}+N\mat{1+\frac{\ell}{2}b'}{1+\frac{\ell}{2}b'}{0}{1+\frac{\ell}{2}b'} .
\end{align*}
\end{proof}

\begin{proposition}\label{3.5}
The following are equivalent:
\begin{enumerate}
\item[(i)] $T^{k}R^{\ell}(R^{\ell}T^{k})^{-1}\in K$  holds for all $k,\ell\in \N$ such that $k$ is odd, $\ell$ is even, $N=k\ell$ and $\gcd(k,\ell)=1$. 
\item[(ii)]  $u'\equiv_2\begin{cases} 1+b' & \text{if}\ N\equiv_4 2,\\ 1 & \text{if}\ N\equiv_4 0  .\end{cases}$
\end{enumerate}
\end{proposition}
\begin{proof}
Let $k\in\N$ be odd and $\ell\in\N$ be even such that $N=k\ell$ and $\gcd(k,\ell)=1$. 
By Lemmas \ref{N} and \ref{odd_even} and the definition of the group $K$, the condition
$$T^k\cdot R^\ell\left(R^\ell\cdot T^k\right)^{-1}=\left(A^kB^\ell(B^\ell A^k)^{-1},\vc{\frac{N}{2}(\ell-1)u'}{\frac{N}{2}(k-1)v}\right)=$$
$$=\left(\mat{1}{0}{0}{1}+N\mat{1+\frac{\ell}{2}b'}{1+\frac{\ell}{2}b'}{0}{1+\frac{\ell}{2}b'},\vc{\frac{N}{2}u'}{0}\right)\in K$$
is equivalent to 
$u'\equiv_2 1+\frac{\ell}{2}b'$. The rest is easy.
\end{proof}

\section{Condition $(T^k R^{\ell}T^k)^2(TRT)^{-2}\in K$ for $0\leq k,\ell\leq N-1$ and $k\ell\equiv_N 1$}\label{section_10}

%\section{Condition $(\tilde{t}^{k}\tilde{r}^{\ell}\tilde{t}^k)^2=(\tilde{t}\tilde{r}\tilde{t})^2$ for $0\leq k,\ell\leq N-1$ and $k\ell\equiv_N 1$}

In this section we assume that $1\leq k,\ell\leq N-1$ and $k\ell\equiv_N 1$. It follows that both $k$ and $\ell$ are odd.

\begin{lemma}\label{4.1}
Let  $\f(T)=h=\left(\begin{smallmatrix}u\\v\end{smallmatrix}\right)\in\Z_N^2$ and $\f(R)=h'=\left(\begin{smallmatrix}u'\\v'\end{smallmatrix}\right)\in\Z_N^2$. Then:
\begin{enumerate}
\item[(i)] $\f(T^k R^\ell T^k)=\matN{k}{\frac{k(3k-1)}{2}}{-1}{\frac{3k+1}{2}}\cdot h+\matN{\frac{\ell+1}{2}}{1}{-{\ell\choose 2}}{\ell}\cdot h'$,

%(ii) $$(A,h)(B,h')(A,h)=\left(ABA,\left[\mat{1}{1}{-1}{2}\right]_N\cdot h+\left[\mat{1}{1}{0}{1}\right]_N\cdot h'\right)$$

\item[(ii)] $\f\left(\left(T^k R^\ell T^k\right)^{2}\right)= \matN{0}{3k^2}{-2}{1}\cdot h+\matN{1}{2}{-\ell^2}{0}\cdot h'$,

\item[(iii)] $\f\left(\left(TRT\right)^{2}\right)=\matN{0}{3}{-2}{1}\cdot h+\matN{1}{2}{-1}{0}\cdot h'$,

\item[(iv)] $\f\left(\left(T^k R^\ell T^k\right)^{2}\left(TRT\right)^{-2}\right)= \vc{3(k^2-1)v}{-(\ell^2-1)u'}$.
$$$$
\end{enumerate}
\end{lemma}
\begin{proof}
 (i) Firstly, $$\left[A^k B^{\ell}\right]_N=\matN{1}{k}{0}{1}\matN{1}{0}{-\ell}{1}=\matN{0}{k}{-\ell}{1}.$$

 Since $kl\equiv_{N}1$, it follows that $\ell-k{\ell\choose 2}=\ell-k\ell\frac{\ell-1}{2}\equiv_{N}\ell-\frac{\ell-1}{2}=\frac{\ell+1}{2}$ and, symmetrically, $k-\ell{k\choose 2}\equiv_{N}\frac{k+1}{2}$ . Hence, by Lemma \ref{product}(i) and (ii), we obtain
 \begin{align*}
 &\f(T^k R^\ell T^k)=\f(T^k R^\ell)+\left[\g(T^k R^{\ell})\right]_N \cdot\f(T^k)= \f(T^k R^\ell)+\left[A^k B^{\ell}\right]_N \cdot\f(T^k)=\\
 &=\matN{k}{{k\choose 2}}{0}{k}\cdot h+\matN{\frac{\ell+1}{2}}{1}{-{\ell\choose 2}}{\ell}\cdot h'+\matN{0}{k}{-\ell}{1}\matN{k}{{k\choose 2}}{0}{k}\cdot h=\\
 &=\matN{k}{{k\choose 2}}{0}{k}\cdot h+\matN{\frac{\ell+1}{2}}{1}{-{\ell\choose 2}}{\ell}\cdot h'+\matN{0}{k^2}{-1}{\frac{k+1}{2}}\cdot h=\\
  &=\matN{k}{\frac{k(3k-1)}{2}}{-1}{\frac{3k+1}{2}}\cdot h+\matN{\frac{\ell+1}{2}}{1}{-{\ell\choose 2}}{\ell}\cdot h'\ .
  \end{align*}
  (ii) By the proof of the part (i), we have $$\left[A^k B^{\ell}A^k\right]_N=\left[A^k B^{\ell}]_{N}\cdot[A^k\right]_N=\matN{0}{k}{-\ell}{1}\matN{1}{k}{0}{1}=\matN{0}{k}{-\ell}{0}$$
  and,  by the part (i), $$w=\f(T^k R^{\ell}T^k)=\matN{k}{\frac{k(3k-1)}{2}}{-1}{\frac{3k+1}{2}}\cdot h+\matN{\frac{\ell+1}{2}}{1}{-{\ell\choose 2}}{\ell}\cdot h'.$$
  Hence we obtain
  \begin{align*}
  &\f((T^k R^{\ell}T^k)^2)=\f((T^k R^{\ell}T^k)(T^k R^{\ell}T^k))=w+[A^k B^{\ell}A^k]_N\cdot w=\\
  &=w+\matN{0}{k}{-\ell}{0}\cdot w=\matN{1}{k}{-\ell}{1}\cdot w=\\
  &=\matN{1}{k}{-\ell}{1}\matN{k}{\frac{k(3k-1)}{2}}{-1}{\frac{3k+1}{2}}\cdot h+\matN{1}{k}{-\ell}{1}\matN{\frac{\ell+1}{2}}{1}{-{\ell\choose 2}}{\ell}\cdot h'=\\
  &=\matN{0}{3k^2}{-2}{1}\cdot h+\matN{1}{2}{-\ell^2}{0}\cdot h'
  \end{align*}
where we have used the equalities $-\ell k\frac{3k-1}{2}+\frac{3k+1}{2}\equiv_N-\frac{3k-1}{2}+\frac{3k+1}{2}=1$ and
$\frac{\ell+1}{2}-k{\ell\choose 2}=\frac{\ell+1}{2}-k\ell\frac{\ell-1}{2}\equiv_{N}\frac{\ell+1}{2}-\frac{\ell-1}{2}=1$.

(iii) Follows immediately from (ii).

(iv) By the generating relations for 
$t=\left[\begin{smallmatrix}1 & 1\\ 0 & 1\end{smallmatrix}\right]_{N}=[A]_N$ and $r=\left[\begin{smallmatrix}1 & 0\\ -1 & 1\end{smallmatrix}\right]_{N}=[B]_N$ of the group $\SL(2,\Z_{N})$ it holds that $$\left[\g\left(\left(T^k R^\ell T^k\right)^{2}\right)\right]_{N}=\left[(A^kB^\ell A^k)^2\right]_N=(t^kr^{\ell}t^k)^2=(trt)^2=$$
$$=\left[(AB A)^2\right]_N=\left[\g\left(\left(T RT\right)^{2}\right)\right]_{N} .$$
Hence, by the parts (ii), (iii)  and by Lemma \ref{basic_0}(iii), we have 
\begin{align*}
&\f\Big(\left(T^k R^\ell T^k\right)^{2}\left(TRT\right)^{-2}\Big)=\f\Big(\left(T^k R^\ell T^k\right)^{2}\Big)-\f\Big(\left(TRT\right)^{2}\Big)=\\
&=\matN{0}{3k^2}{-2}{1}\cdot h+\matN{1}{2}{-\ell^2}{0}\cdot h'-\matN{0}{3}{-2}{1}\cdot h-\matN{1}{2}{-1}{0}\cdot h'=\\
&=\matN{0}{3(k^2-1)}{0}{0}\cdot h+\matN{0}{0}{-\ell^2+1}{0}\cdot h'=\\
&=\matN{0}{3(k^2-1)}{0}{0}\cdot \vc{v}{u}+\matN{0}{0}{-\ell^2+1}{0}\cdot \vc{v'}{u'}=\vc{3(k^2-1)v}{-(\ell^2-1)u'}.
\end{align*}
\end{proof}

\begin{lemma}\label{10.2}
Let $k\ell=1+\varepsilon N$, where $\varepsilon\in\N_0$. Then:
\begin{enumerate}
\item[(i)] $A^k B^{\ell}=$
$$=\mat{0}{k}{-\ell}{1}+
 N\mat{b+\frac{k+1}{2}c+\frac{\ell-1}{2}b'+c'+\varepsilon}{a+b+a'+\frac{\ell-1}{2}b'}{a+\frac{k+1}{2}c+a'+c'}{a+\frac{k-1}{2}c+a'+\frac{\ell+1}{2}b'},$$

\item[(ii)] $A^k B^{\ell}A^k=$
 $$=\mat{0}{k}{-\ell}{0}+
 N\mat{b+\frac{k-1}{2}c+\frac{\ell-1}{2}b'+c'+\varepsilon}{c+a'+c'+\varepsilon}{c+a'+c'}{b+\frac{k+1}{2}c+\frac{\ell+1}{2}b'+c'+\varepsilon},$$
 \item[(iii)] $(A^k B^{\ell}A^k)^2=\mat{-1}{0}{0}{-1}+
 N\mat{0}{c+b'}{c+b'}{0}=(ABA)^2$,
    
\item[(iv)] $(A^k B^{\ell}A^k)^2(ABA)^{-2}=\mat{1}{0}{0}{1}$.
\end{enumerate}
\end{lemma}
\begin{proof}
(i) Since both $k$ and $\ell$ are odd, it follows that $N\cdot(\begin{smallmatrix}1 & k\\ 0 & 1\end{smallmatrix})=N\cdot(\begin{smallmatrix}1 & 1\\ 0 & 1\end{smallmatrix})$ and $N\cdot(\begin{smallmatrix}1 & 0\\ -\ell & 1\end{smallmatrix})=N\cdot(\begin{smallmatrix}1 & 0\\ 1 & 1\end{smallmatrix})$ when computing modulo $2N$. Using these equalities, the equality $k\ell=1+\varepsilon N$, for some $\varepsilon\in\N_{0}$, and Corollary \ref{7.6} we obtain
\begin{align*}
 &A^k B^{\ell}=\Bigg(\mat{1}{k}{0}{1}+N\mat{a+\frac{k+1}{2}c}{a+b}{c}{a+\frac{k-1}{2}c}\Bigg)\cdot\\
 &\hspace{40mm}\cdot\Bigg(\mat{1}{0}{-\ell}{1}+N\mat{a'+\frac{\ell-1}{2}b'}{b'}{a'+c'}{a'+\frac{\ell+1}{2}b'}\Bigg)=\\
 &=\mat{-\varepsilon N}{k}{-\ell}{1}+N\mat{1}{1}{0}{1}\mat{a'+\frac{\ell-1}{2}b'}{b'}{a'+c'}{a'+\frac{\ell+1}{2}b'}+\\
&\hspace{55mm}+ N\mat{a+\frac{k+1}{2}c}{a+b}{c}{a+\frac{k-1}{2}c}\mat{1}{0}{1}{1}=\\
 &=\mat{-\varepsilon N}{k}{-\ell}{1}+N\mat{\frac{\ell-1}{2}b'+c'}{a'+\frac{\ell+3}{2}b'}{a'+c'}{a'+\frac{\ell+1}{2}b'}+
 N\mat{b+\frac{k+1}{2}c}{a+b}{a+\frac{k+1}{2}c}{a+\frac{k-1}{2}c}
 \end{align*}
 
 (ii) By the part (i) and by using again the equality $-N\equiv_{2N}N$ and the equality $N\cdot(\begin{smallmatrix}0 & k\\ -\ell & 1\end{smallmatrix})=N\cdot(\begin{smallmatrix}0 & 1\\ 1 & 1\end{smallmatrix})$, we obtain
 \begin{align*}
 &A^k B^{\ell}A^k=(A^k B^{\ell})A^k=\\
 &=\Bigg(\mat{0}{k}{-\ell}{1}+
 N\mat{b+\frac{k+1}{2}c+\frac{\ell-1}{2}b'+c'+\varepsilon}{a+b+a'+\frac{\ell-1}{2}b'}{a+\frac{k+1}{2}c+a'+c'}{a+\frac{k-1}{2}c+a'+\frac{\ell+1}{2}b'}\Bigg)\cdot\\
 &\hspace{45mm}\cdot\Bigg(\mat{1}{k}{0}{1}+N\mat{a+\frac{k+1}{2}c}{a+b}{c}{a+\frac{k-1}{2}c}\Bigg)=\\
  &=\mat{0}{k}{-\ell}{-\varepsilon N}+
 N\mat{0}{1}{1}{1}\mat{a+\frac{k+1}{2}c}{a+b}{c}{a+\frac{k-1}{2}c}+\\
 &\hspace{10mm}+N\mat{b+\frac{k+1}{2}c+\frac{\ell-1}{2}b'+c'+\varepsilon}{a+b+a'+\frac{\ell-1}{2}b'}{a+\frac{k+1}{2}c+a'+c'}{a+\frac{k-1}{2}c+a'+\frac{\ell+1}{2}b'}\mat{1}{1}{0}{1}=\\
 &=\mat{0}{k}{-\ell}{-\varepsilon N}
+N\mat{c}{a+\frac{k-1}{2}c}{a+\frac{k-1}{2}c}{b+\frac{k-1}{2}c}+\\
 &\hspace{28mm}+ N\mat{b+\frac{k+1}{2}c+\frac{\ell-1}{2}b'+c'+\varepsilon}{a+\frac{k+1}{2}c+a'+c'+\varepsilon}{a+\frac{k+1}{2}c+a'+c'}{c+\frac{\ell+1}{2}b'+c'}=\\
 &=\mat{0}{k}{-\ell}{-\varepsilon N}+N\mat{\alpha}{\beta}{\gamma}{\delta}
 \end{align*}
 where $$\mat{\alpha}{\beta}{\gamma}{\delta}=\mat{b+\frac{k-1}{2}c+\frac{\ell-1}{2}b'+c'+\varepsilon}{c+a'+c'+\varepsilon}{c+a'+c'}{b+\frac{k+1}{2}c+\frac{\ell+1}{2}b'+c'+\varepsilon}\ .$$

 (iii) By the proof of part (ii) and  with help of the equality $N\cdot(\begin{smallmatrix}0 & k\\ -\ell & 0\end{smallmatrix})=N\cdot(\begin{smallmatrix}0 & 1\\ 1 & 0\end{smallmatrix})$ it follows that 
 \begin{align*}
 (A^k B^{\ell}A^k)^2&=\Bigg(\mat{0}{k}{-\ell}{0}+
 N\mat{\alpha}{\beta}{\gamma}{\delta}\Bigg)^2=\mat{-1-\varepsilon N}{0}{0}{-1-\varepsilon N}+\\
 &+N\mat{0}{1}{1}{0}\mat{\alpha}{\beta}{\gamma}{\delta}+N\mat{\alpha}{\beta}{\gamma}{\delta}\mat{0}{1}{1}{0}=\\
 &=\mat{-1-\varepsilon N}{0}{0}{-1-\varepsilon N}+N\mat{\gamma}{\delta}{\alpha}{\beta}+N\mat{\beta}{\alpha}{\delta}{\gamma}=\\
 &=\mat{-1}{0}{0}{-1}+N\mat{\beta+\gamma-\varepsilon}{\alpha+\delta}{\alpha+\delta}{\beta+\gamma-\varepsilon}=\\
 &=\mat{-1}{0}{0}{-1}+
 N\mat{0}{c+b'}{c+b'}{0}\ .
 \end{align*}
 This expression does not depend on $k$ and $\ell$. Hence we obtain $(A^{k}B^{\ell}A^k)^2=(ABA)^2$.
 
 (iv) Follows immediately from (iii).
\end{proof}

\begin{remark}\label{10.3}
 The following are equivalent for $e\in\N_0$:
 \begin{enumerate}
  \item There is $s\in\{0,1\}$ such that $e\equiv_{N}\frac{N}{2}s$ and $Ns\equiv_{2N}0$,
  \item $e\equiv_{N}0$.
 \end{enumerate}
\end{remark}

\begin{proposition}\label{4.3}
 Let $1\leq k,\ell\leq N-1$ be  such that $k\ell\equiv_N 1$. Then the condition $(T^k R^{\ell}T^k)^2(TRT)^{-2}\in K$  holds if and only if  $3(k^2-1)v\equiv_N 0$ and 
$(\ell^2-1)u'\equiv_N 0$.
\end{proposition}
\begin{proof}
By Lemmas \ref{4.1}(iv), \ref{10.2}(iv), the definition of the group $K$ and by  Remark \ref{10.3}, the condition
$$\left(T^k R^\ell T^k\right)^{2}\left(TRT\right)^{-2}=\left((A^k B^{\ell}A^k)^2(A BA)^{-2}, \vc{3(k^2-1)v}{-(\ell^2-1)u'}\right)=$$
$$=\left(\mat{1}{0}{0}{1},\vc{3(k^2-1)v}{-(\ell^2-1)u'}\right)\in K$$
is equivalent to conditions
\begin{align*}
3(k^2-1)v&\equiv_N 0 \\
(\ell^2-1)u'&\equiv_N 0\ .
\end{align*}
%In particular, if $\ell=1$ then it follows from $(\ell^2+1)u'\equiv_N 0$ that  $u'\equiv_N \frac{N}{2}$. On the other hand, if $u'\equiv_N \frac{N}{2}$, then, as $\ell$ is odd, we obtain that $(\ell^2+1)u'\equiv_N 0$.
\end{proof}

\section{Condition $T^k R^{\ell}T^k(R^{\ell}T^k R^{\ell})^{-1}\in K$ for $1\leq k,\ell\leq N-1$ and $k\ell\equiv_N 1$}\label{section_11}

%\section{Condition $\tilde{t}^{k}\tilde{r}^{\ell}\tilde{t}^k=\tilde{r}^{\ell}\tilde{t}^{k}\tilde{r}^{\ell}$ for $0\leq k,\ell\leq N-1$ and $k\ell\equiv_N 1$}

In this section we assume again that $1\leq k,\ell\leq N-1$ and $k\ell\equiv_N 1$. It follows that both $k$ and $\ell$ are odd.

\begin{lemma}\label{5.1}
Let $\f(T)=h=\left(\begin{smallmatrix}u\\v\end{smallmatrix}\right)\in\Z_N^2$ and $\f(R)=h'=\left(\begin{smallmatrix}u'\\v'\end{smallmatrix}\right)\in\Z_N^2$. Then:
\begin{enumerate}
\item[(i)] $\f(R^\ell T^k R^\ell)=\matN{k}{{k\choose 2}}{-1}{\frac{k+1}{2}}\cdot h+\matN{\frac{3\ell+1}{2}}{1}{-\frac{3\ell-1}{2}\ell}{\ell}\cdot h'$,

%(ii) $$(A,h)(B,h')(A,h)=\left(ABA,\left[\mat{1}{1}{-1}{2}\right]_N\cdot h+\left[\mat{1}{1}{0}{1}\right]_N\cdot h'\right)$$

\item[(ii)] $\f\left(T^k R^\ell T^k\left(R^\ell T^k R^\ell\right)^{-1}\right)=
\vc{k^2v-\ell u'}{kv+\ell^2 u'}$.
$$$$
\end{enumerate}
\end{lemma}
\begin{proof}
 (i) Since $k\ell\equiv_N 1$, we have $\ell-k{\ell\choose 2}=\ell-k\ell\frac{\ell-1}{2}\equiv_{N} \ell-\frac{\ell-1}{2}= \frac{\ell+1}{2}$ and $-\ell{k\choose 2}+k=-\ell k\frac{k-1}{2}+k\equiv_{N} -\frac{k-1}{2}+k= \frac{k+1}{2}$. Further, $-\ell\frac{(\ell+1)}{2}-{\ell\choose 2}= -\frac{\ell^2+\ell}{2}-\frac{\ell^2-\ell}{2}=-\ell^2.$ Hence, by Lemma \ref{product}(i) and (ii), we obtain
 \begin{align*}
 &\f(R^\ell T^k R^\ell)=\f(R^\ell)+[B^\ell]_N \cdot\f(T^k R^{\ell})=\f(R^\ell)+\matN{1}{0}{-\ell}{1} \cdot\f(T^k R^{\ell})=\\
 &=\matN{\ell}{0}{-{\ell\choose 2}}{\ell}\cdot h'+\matN{1}{0}{-\ell}{1}\cdot\left(\matN{k}{{k\choose 2}}{0}{k}\cdot h+\matN{\frac{\ell+1}{2}}{1}{-{\ell\choose 2}}{\ell}\cdot h'\right)=\\
 &=\matN{\ell}{0}{-{\ell\choose 2}}{\ell}\cdot h'+\matN{k}{{k\choose 2}}{-1}{\frac{k+1}{2}}\cdot h+\matN{\frac{\ell+1}{2}}{1}{-\ell^2}{0}\cdot h'=\\
&=\matN{k}{{k\choose 2}}{-1}{\frac{k+1}{2}}\cdot h+\matN{\frac{3\ell+1}{2}}{1}{-\frac{\ell(3\ell-1)}{2}}{\ell}\cdot h'.
\end{align*}
  
  (ii)  By the generating relations for 
$t=\left[\begin{smallmatrix}1 & 1\\ 0 & 1\end{smallmatrix}\right]_{N}=[A]_N$ and $r=\left[\begin{smallmatrix}1 & 0\\ -1 & 1\end{smallmatrix}\right]_{N}=[B]_N$ of the group $\SL(2,\Z_{N})$ it holds that $$[\g(T^k R^\ell T^k)]_{N}=\left[A^kB^\ell A^k\right]_N=t^kr^{\ell}t^k=r^{\ell}t^{k}r^{\ell}=\left[B^{\ell} A^k B^{\ell}\right]_N=\left[\g\left(R^\ell T^kR^\ell\right)\right]_N\ .$$ 
Hence, by the part (i) and Lemmas \ref{basic_0}(iii) and \ref{4.1}, we have
\begin{align*}
&\f\Big(\left(T^k R^\ell T^k\right)\left(R^\ell T^kR^\ell\right)^{-1}\Big)=\f\Big(T^k R^\ell T^k\Big)-\f\Big(R^\ell T^kR^\ell\Big)=\\
&=\matN{k}{\frac{k(3k-1)}{2}}{-1}{\frac{3k+1}{2}}\cdot h+\matN{\frac{\ell+1}{2}}{1}{-{\ell\choose 2}}{\ell}\cdot h'\\
&\quad-\matN{k}{{k\choose 2}}{-1}{\frac{k+1}{2}}\cdot h-\matN{\frac{3\ell+1}{2}}{1}{-\frac{\ell(3\ell-1)}{2}}{\ell}\cdot h'=\\
&=\matN{0}{k^2}{0}{k}\cdot h+\matN{-\ell}{0}{\ell^2}{0}\cdot h'=\\
&=\matN{0}{k^2}{0}{k}\cdot \vc{u}{v}+\matN{-\ell}{0}{\ell^2}{0}\cdot \vc{u'}{v'}=\vc{k^2v-\ell u'}{kv+\ell^2 u'}.
\end{align*}
\end{proof}

\begin{lemma}\label{5.2}
Let $k\ell=1+\varepsilon N$, where $\varepsilon\in\N_0$. Then $$A^k B^{\ell}A^k(B^{\ell}A^k B^{\ell})^{-1}=\mat{1}{0}{0}{1}+
 N\mat{r}{c+b'}{c+b'}{r}$$
where $r=a+b+c+a'+b'+c'+\varepsilon$.
\end{lemma}
\begin{proof}
Since both $k$ and $\ell$ are odd, it follows that $N\cdot(\begin{smallmatrix}1 & 0\\ -\ell & 1\end{smallmatrix})=N\cdot(\begin{smallmatrix}1 & 0\\ 1 & 1\end{smallmatrix})$ and $N\cdot(\begin{smallmatrix}0 & k\\ -\ell & 1\end{smallmatrix})=N\cdot(\begin{smallmatrix}0 & 1\\ 1 & 1\end{smallmatrix})$ when computing modulo $2N$. Using these equalities, the equality $k\ell=1+\varepsilon N$, for some $\varepsilon\in\N_{0}$, and Corollary \ref{7.6} we obtain
\begin{align*}
 &B^{\ell}A^k B^{\ell}=B^{\ell}(A^k B^{\ell})=\Bigg(\mat{1}{0}{-\ell}{1}+N\mat{a'+\frac{\ell-1}{2}b'}{b'}{a'+c'}{a'+\frac{\ell+1}{2}b'}\Bigg)\cdot\\
 &\cdot\Bigg(\mat{0}{k}{-\ell}{1}+
 N\mat{b+\frac{k+1}{2}c+\frac{\ell-1}{2}b'+c'+\varepsilon}{a+b+a'+\frac{\ell-1}{2}b'}{a+\frac{k+1}{2}c+a'+c'}{a+\frac{k-1}{2}c+a'+\frac{\ell+1}{2}b'}\Bigg)=\\
 & =\mat{0}{k}{-\ell}{-\varepsilon N}+
N\mat{a'+\frac{\ell-1}{2}b'}{b'}{a'+c'}{a'+\frac{\ell+1}{2}b'}\mat{0}{1}{1}{1}+\\
&\hspace{14mm}+N\mat{1}{0}{1}{1}\mat{b+\frac{k+1}{2}c+\frac{\ell-1}{2}b'+c'+\varepsilon}{a+b+a'+\frac{\ell-1}{2}b'}{a+\frac{k+1}{2}c+a'+c'}{a+\frac{k-1}{2}c+a'+\frac{\ell+1}{2}b'}=\\
&=\mat{0}{k}{-\ell}{-\varepsilon N}+
 N\mat{b'}{a'+\frac{\ell+1}{2}b'}{a'+\frac{\ell+1}{2}b'}{\frac{\ell+1}{2}b'+c'}+\\
 &\hspace{37mm}+N\mat{b+\frac{k+1}{2}c+\frac{\ell-1}{2}b'+c'+\varepsilon}{a+b+a'+\frac{\ell-1}{2}b'}{a+b+a'+\frac{\ell-1}{2}b'+\varepsilon}{b+\frac{k-1}{2}c+b'}=\\
 &=\mat{0}{k}{-\ell}{0}+
 N\mat{\alpha}{\beta}{\gamma}{\delta}.
 \end{align*}
 where 
 $$\mat{\alpha}{\beta}{\gamma}{\delta}=\mat{b+\frac{k+1}{2}c+\frac{\ell+1}{2}b'+c'+\varepsilon}{a+b+b'}{a+b+b'+\varepsilon}{b+\frac{k-1}{2}c+\frac{\ell-1}{2}b'+c'+\varepsilon}\ .$$
 Hence, by Remark \ref{inverse}, it is
 $$(B^{\ell}A^k B^{\ell})^{-1}=\mat{0}{-k}{\ell}{0}+
 N\mat{\delta}{\beta}{\gamma}{\alpha}.$$
 Finally, using $N\cdot(\begin{smallmatrix}0 & k\\ -\ell & 0\end{smallmatrix})=N\cdot(\begin{smallmatrix}0 & 1\\ 1 & 0\end{smallmatrix})=N\cdot(\begin{smallmatrix}0 & -k\\ -\ell & 0\end{smallmatrix})$, Lemma \ref{10.2}(ii) and setting $r=a+b+c+a'+b'+c'+\varepsilon$,  we obtain
 \begin{align*}
  &A^k B^{\ell}A^k(B^{\ell}A^k B^{\ell})^{-1}
 =\\
 &=\Bigg(\mat{0}{k}{-\ell}{0}+
 N\mat{\delta}{r-\beta}{r-\gamma}{\alpha}\Bigg)
 \cdot\Bigg(\mat{0}{-k}{\ell}{0}+
 N\mat{\delta}{\beta}{\gamma}{\alpha}\Bigg)=\\
 &=\mat{1+\varepsilon N}{0}{0}{1+\varepsilon N}
+N\mat{\delta}{r-\beta}{r-\gamma}{\alpha}\mat{0}{1}{1}{0}+\\
&\hspace{80mm}+N\mat{0}{1}{1}{0}\mat{\delta}{\beta}{\gamma}{\alpha}=\\
 &=\mat{1+\varepsilon N}{0}{0}{1+\varepsilon N}
 +N\mat{r-\beta}{\delta}{\alpha}{r-\gamma}
+N\mat{\gamma}{\alpha}{\delta}{\beta}=\\
&=\mat{1}{0}{0}{1}
 +N\mat{r-\beta+\gamma+\varepsilon}{\alpha+\delta}{\alpha+\delta}{r+\beta-\gamma+\varepsilon}=\\
 &=\mat{1}{0}{0}{1}+
 N\mat{r}{c+b'}{c+b'}{r}.
 \end{align*}
\end{proof}

\begin{proposition}\label{5.3}
The following conditions are equivalent:
\begin{enumerate}
\item[(i)] $T^k R^{\ell}T^k(R^{\ell}T^k R^{\ell})^{-1}\in K$ holds for all $1\leq k,\ell\leq N-1$  such that $k\ell\equiv_N 1$.   
\item[(ii)] $v=\frac{N}{2}\beta$, $u'=\frac{N}{2}\alpha'$ for some $\alpha',\beta\in\{0,1\}$ and $\alpha'+\beta\equiv_2 c+b'$.
\end{enumerate}
\end{proposition}
\begin{proof}
Let $k,\ell\in\{1,\dots,N-1\}$  be such that $k\ell\equiv_{2N}1+\varepsilon N$ for some $\varepsilon\in\{0,1\}$.
By  Lemmas \ref{5.1}(ii) and \ref{5.2}  we have
$$T^k R^\ell T^k\left(R^\ell T^k R^\ell\right)^{-1}=\left(A^k B^{\ell}A^k(B^{\ell} A^k B^{\ell})^{-1}, 
\vc{k^2v-\ell u'}{kv+\ell^2 u'}\right)=$$
$$=\left(
 \mat{1+Nr}{N(c+b')}{N(c+b')}{1+Nr},\vc{k^2v-\ell u'}{kv+\ell^2 u'}\right)$$
 where $r=a+b+c+a'+b'+c'+\varepsilon$.

Assume the condition (i). For $k=\ell=1$ we obtain, by the definition of the group $K$, that $v-u'\equiv_{N/2}0$ and $v+u'\equiv_{N/2}0$. Hence $v-u'=\gamma\frac{N}{2}$ and $v+u'=\varepsilon\frac{N}{2}$ for some $\gamma, \varepsilon\in\Z$. The definition of the group $K$ now implies that $\gamma\equiv_2 c+b'\equiv_2\varepsilon$. It follows that $\gamma+\varepsilon\equiv_2 \gamma-\varepsilon\equiv_2 0$ and therefore we have $\gamma+\varepsilon=2\beta$ and  $\gamma-\varepsilon=2\alpha'$ for some $\alpha',\beta\in\Z$. Finally, $v=\frac{(\gamma+\varepsilon)N}{4}=\beta\frac{N}{2}$, $u'=\frac{(\gamma-\varepsilon)N}{4}=\alpha'\frac{N}{2}$  and $\alpha'+\beta=\gamma\equiv_2 c+b'$. Now we can assume without loss of generality that $\alpha',\beta\in\{0,1\}$. Hence, we have obtained the condition (2).

Assume the condition (ii). With help of the fact that both $k$ and $\ell$ are odd numbers, we have $k^2v-\ell u'\equiv_{N}v+u'=\frac{N}{2}(\alpha'+\beta)\equiv_{2}\frac{N}{2}(c+b')$ and, similarly, $kv+\ell^2 u'\equiv_{N}v+u'=\frac{N}{2}(\alpha'+\beta)\equiv_{2}\frac{N}{2}(c+b')$. By the definition of the group $K$ it is now clear that 
$T^k R^\ell T^k\left(R^\ell T^k R^\ell\right)^{-1}\in K$.
%In particular, if $\ell=1$ then it follows from $(\ell^2+1)u'\equiv_N 0$ that  $u'\equiv_N \frac{N}{2}$. On the other hand, if $u'\equiv_N \frac{N}{2}$, then, as $\ell$ is odd, we obtain that $(\ell^2+1)u'\equiv_N 0$.
\end{proof}

\begin{proposition}\label{6.3}
 Let $N\equiv_4 2$ and $T,R\in \SL(2,\Z_{2N})\ltimes \Z_{N}^2$ have the form as in Section $\ref{section_7}$. Then the following are equivalent:
 \begin{enumerate}
  \item[(i)] $T$ and $R$ fulfill the conditions (i)--(v) in Proposition $\ref{character_2}$.
  \item[(ii)] $a'\equiv_2 c'$, $v\equiv_N\frac{N}{2}(c+1)$ and $u'\equiv_N\frac{N}{2}(b'+1)$.
 \end{enumerate}
\end{proposition}
\begin{proof}
Since $N\equiv_4 2$, the number $\frac{N}{2}$ is odd.

 Assume the condition (i). By Proposition \ref{3.3}, $a'\equiv_2 c'$ and $v\equiv_2 c+1$. By Proposition \ref{5.3}, $v=\frac{N}{2}\beta$ for some $\beta\in\N_0$. Hence $\beta\equiv_2 \frac{N}{2}\beta=v\equiv_2 c+1$. From $\beta\equiv_2 (c+1)$ it follows that $v=\frac{N}{2}\beta\equiv_{N}\frac{N}{2}(c+1)$.
 
 Similarly, we obtain $u'\equiv_N \frac{N}{2}(b'+1)$ by Propositions \ref{3.3} and \ref{5.3}.
 
  Assume the condition (ii). Then $v\equiv_2\frac{N}{2}(c+1)\equiv_2 c+1\equiv_2\frac{N}{2}c+1$ and, similarly, we obtain $u'\equiv_2\frac{N}{2}b'+1$. Hence the conditions (i) in Propositions \ref{2.2}, \ref{3.3} and \ref{3.5} are fulfilled.
  
  Let $k\in\N$ be an odd number. Then $k^2-1=2m$ for some $m\in\Z$. From $v\equiv_2\frac{N}{2}(c+1)$ it follows that $3(k^2-1)v=6mv\equiv_N 6m\frac{N}{2}(c+1)\equiv_N 0$. Similarly, we obtain $(\ell^2-1)u'\equiv_N 0$ for $\ell$ odd. In this way condition (i) in Proposition \ref{4.3} are fulfilled.
  
  Finally, as $(c+1)+(b'+1)\equiv_2 c+b'$, the condition (i) in Proposition \ref{5.3} is fulfilled. 
\end{proof}

\section{The main results}\label{section_12}

\begin{theorem}\label{6.1}
 Let $N\equiv_4 0$. Then the projective Clifford group $\overline{C}(N)$ is not a semidirect product corresponding to the sequence
 $$1 \to \overline{H}(N)\to \overline{C}(N)\to \overline{C}(N)/\overline{H}(N)\to 1
 $$   (or equivalently, to the exact sequence $(\ref{exact1})$). Consequently, none of these two exact sequences is right splitting.  
\end{theorem}
\begin{proof}
%  The sequence from the statement is isomorphic to the sequence $(\ref{exact1})$)
%  $$1 \to \Z_{N}^2\to \left(\SL(2,\Z_{2N})\ltimes \Z_{N}^2\right)/K \to 
%  \SL(2,\Z_N) \to 1$$
 Assume, on the contrary, that the projective Clifford group $\overline{C}(N)$ is a semidirect product corresponding to any of these mutually equivalent exact sequences. Then, by Proposition \ref{character_2}, there are $T,R\in \SL(2,\Z_{2N})\ltimes \Z_{N}^2$ fulfilling the conditions (i)-(v) in Proposition \ref{character_2}. Hence, by Propositions \ref{3.3} and \ref{5.3}, we have $v\equiv_2 1$ and $v=\frac{N}{2}\beta$ for some $\beta\in\N_0$. From $N\equiv_4 0$ it follows that $v=\frac{N}{2}\beta\equiv_2 0$, a contradiction. Therefore the projective Clifford group $\overline{C}(N)$ is not a semidirect product of the proposed form and the corresponding exact sequences are not right splitting.
\end{proof}

\begin{corollary}\label{6.2}
  Let $N\equiv_4 0$. Then the  Clifford group $C(N)$ is not a semidirect product corresponding to the sequence
 $$1 \to H(N)\to C(N)\to C(N)/H(N)\to 1
 $$   (or equivalently, to the exact sequence $(\ref{exact2})$). Consequently, none of these two exact sequences is right splitting.  
\end{corollary}
\begin{proof}
 Consider the commutative diagram
 \begin{equation*} 
\begin{tikzcd}\label{diagram_2}
1 \arrow{r}{}& H(N) \arrow{r}{\iota} & C(N) \arrow{d}{p} \arrow{r}{\pi}  & \SL(2,\Z_{N}) \arrow{d}{id_{\SL(2,\Z_N)}} \arrow{r}{} & 1\\
1 \arrow{r}{} & \Z_N^2 \arrow{r}{\nu} &\overline{C}(N) \arrow{r}{\overline{\pi}}  & \SL(2,\Z_{N}) \arrow{r}{} & 1  & 
\end{tikzcd}
\end{equation*}
% \begin{tikzcd}\label{diagram_2}
% C(N) \arrow{d}{p} \arrow{r}{\pi}  & \SL(2,\Z_{N}) \\%
% \overline{C}(N)  \arrow[swap]{ur}{\overline{\pi}}   & 
% \end{tikzcd}
where $p:C(N)\to\overline{C}(N)$ is the natural projection and $\pi=\overline{\pi}\circ p$.

Assume, for contrary, that the upper sequence (i.e., the sequence $(\ref{exact2})$)
is right splitting. Then, according to Section \ref{semidirect}, there is a group homomorphism $\mu:\SL(2,\Z_{N})\to C(N)$ such that $\pi\circ\mu=id_{\SL(2,\Z_N)}$.
 It follows that $\overline{\pi}\circ (p\circ \mu)=\pi\circ\mu=id_{\SL(2,\Z_N)}$ and the lower exact sequence (i.e., the sequence (\ref{exact1})) is therefore right splitting with the homomorphism $p\circ\mu:\SL(2,\Z_{N})\to \overline{C}(N)$, a contradiction with Theorem \ref{6.1}.
\end{proof}

\begin{theorem}\label{6.4}
 Let $N\equiv_4 2$. Then $$
\overline{C}(N) \cong   \SL(2,\mathbb{Z}_N)\ltimes_{\overline{\varphi}} \mathbb{Z}_N^2 
$$
where $\overline{\varphi}(F)(w)=Fw$ for $F\in \SL(2,\mathbb{Z}_N)$ and $w\in\Z_{N}^2$.
\end{theorem}
\begin{proof}
   Set $a=b=a'=c'=0$, $c=b'=1$ and $u=v=u'=v'=0$. Then 
 $$T=(A,h)=\left(\mat{1+N}{1}{N}{1}, 0\right)$$
 $$R=(B,h')=\left(\mat{1}{N}{-1}{1+N}, 0\right)\ .$$
 and the condition (2) in Proposition \ref{6.3} is fulfilled. Hence, by Propositions \ref{6.3} and \ref{character_2}, the group $\overline{C}(N)$ is a semidirect product corresponding to the exact sequence (\ref{exact1}) that is right splitting with some homomorphism $\overline{\mu}:\SL(2,\Z_N)\to \overline{C}(N)$ such that $\overline{\pi}\circ\overline{\mu}=id_{\SL(2,\Z_N)}$. In particular, $
\overline{C}(N) \cong   \SL(2,\mathbb{Z}_N)\ltimes_{\overline{\varphi}} \mathbb{Z}_N^2
$ where, by Section \ref{semidirect}, the map $\overline{\varphi}$ has the form $$\overline{\varphi}(F)(w)=\nu^{-1}\left(\overline{\mu}(F)\cdot \nu(w)\cdot \overline{\mu}(F)^{-1}\right)$$ for all $F\in \SL(2,\mathbb{Z}_N)$ and $w\in \Z_N^2$. Since $\overline{\pi}(\overline{\mu}(F))=F$, we obtain, by the definition of $\overline{\pi}$ in  Theorem \ref{htpp}, that for $\overline{\mu}(F)\in\overline{C}(N)$ it holds that $\overline{\mu}(F)\cdot W(w)_{/\sim}\cdot \overline{\mu}(F)^{-1}=W(Fw)_{/\sim}$. Hence $$\overline{\varphi}(F)(w)=\nu^{-1}\left(\overline{\mu}(F)\cdot \nu(w)\cdot \overline{\mu}(F)^{-1}\right)=\nu^{-1}\left(\overline{\mu}(F)\cdot W(w)_{/\sim}\cdot \overline{\mu}(F)^{-1}\right)=$$
$$=\nu^{-1}\left(W(Fw)_{/\sim}\right)=Fw\ .$$
%  and there is a monomorphism $\nu:SL(2,\Z_N)\to \left(SL(2,\Z_{2N})\ltimes \Z_{N}^2\right)/K$ such that
%  $$\nu\left(\matN{1}{1}{0}{1}\right)=\left(\mat{1+N}{1}{N}{1}, 0\right)_{/K}$$
%  $$\nu\left(\matN{1}{0}{-1}{1}\right)=\left(\mat{1}{N}{-1}{1+N}, 0\right)_{/K}$$
%  and $\pi\circ\nu=id_{SL(2,\Z_N)}$.
\end{proof}

To decide whether the Clifford group $C(N)$ is a semidirect product in the case when $N\equiv_4 2$ seems to be more difficult. This problem is (as in Section \ref{semidirect}) equivalent to the question whether there is a splitting homomorphism $\mu:\SL(2,\Z_N)\to C(N)$. For this task, the group presentation from Section \ref{section_6} can be used again, but at this time a computation with $N\times N$ matrices in $\U(N)$ is needed. Nevertheless, it seems to be  plausible that the answer is positive and such a homomorphism $\mu$ might be constructed as a lifting  of the map $\overline{\mu}$ from the proof of Theorem $\ref{6.4}$.

\section{Conclusions}\label{conclusion}

Our paper brings new results concerning the structure of the projective
Clifford groups in even-dimensional Hilbert spaces.
We were interested in a natural question coming already from classical mechanics:
is the Clifford group a semidirect product?
It is well known that the answer is positive if the dimension of the Hilbert space is odd.
However, the answer was not apparent in even dimensions. 
We turned to the simpler case of projective Clifford groups corresponding to one qudit of dimension $N$
 and proved  that for dimensions divisible by 4 
they are not semidirect products (as it could be expected), but for dimensions $N$ with $N/2$ odd they, surprisingly, are. 
Our proof is  based on an appropriate group presentation of $\SL(2,\Z_N)$ and a simplified description of the projective Clifford group.
However, it is a pity that this approach doesn't allow a direct extension of our results
to non-projective Clifford groups in all even dimensions.
There the situation still remains unclear. 

Based on our results and \cite{bolt1,bolt2} we suggest the following conjecture that concerns the most general case of finite-dimensional quantum mechanics. i.e., when the configuration space is an arbitrary finite abelian group  $A$ (for the details see e.g. \cite{KorbTolar12}).

\

\textbf{Conjecture:} Let the configuration space of a finite-dimensional quantum mechanics be a finite abelian group  $A$ of an even order (i.e., $A=\Z_{n_1}\oplus\cdots\oplus\Z_{n_k}$ and $|A|=n_1\cdots n_k$ is an even number). Then the following are equivalent:
\begin{itemize}
 \item The corresponding Clifford group is a  semidirect product related to its exact sequence with its Heisenberg group.  
 \item The corresponding \emph{projective} Clifford group is a  semidirect product related to its exact sequence with its \emph{projective} Heisenberg group.
 \item $|A|\equiv_{4} 2$.
\end{itemize}

\section*{Acknowledgements}
We remember, with gratitude, the memory of our friend
Ji\v{r}\'{i} Patera, who introduced us into the field of Lie theory --
gradings of semi-simple Lie algebras, especially the Pauli gradings
and their symmetries, now known as the Clifford groups.

We also truly hope that our research helps to lessen 
the language differences that exist between the mathematical 
and the quantum information communities. 
%It is understandable that the quantum information community
%is using some notions differently from the mathematicians.

Both authors are grateful for financial support from RVO14000 and 
”Centre for Advanced Applied Sciences”, 
Registry No.~CZ.02.1.01/0.0/0.0/16\_019/0000778, 
supported by the Operational Programme Research, Development and Education, 
co-financed by the European Structural and Investment Funds.

\end{document}